\documentclass[letterpaper,11pt]{article}

\usepackage{titlesec}
\usepackage[utf8]{inputenc}
\usepackage{graphicx,fullpage,paralist}

\usepackage{setspace}
\usepackage{caption}
\usepackage{amsmath, amssymb, amsthm}
\usepackage{comment,hyperref}
\usepackage{thmtools}
\usepackage{tikz}

\usepackage{pgfplots}
\usepackage{caption}
\usepackage{algorithm} 
\usepackage[algo2e]{algorithm2e} 
\usepackage{lscape}
\usepackage{float}
\usepackage{appendix}
\usepackage{varwidth}
\usepackage{graphicx}
\pgfplotsset{compat=1.10}
\usepgfplotslibrary{fillbetween}
\usetikzlibrary{backgrounds}
\usetikzlibrary{patterns}
\usepackage{natbib}
\usepackage{mathtools}

\usepackage{amsbsy}

\newtheorem{theorem}{Theorem}
\newtheorem{lemma}{Lemma}
\newtheorem{corollary}{Corollary}

\newtheorem*{ra3'}{RA 3'}
\newtheorem*{ra4'}{RA 4'}
\newtheorem*{ra3''}{RA 3''}
\newtheorem*{ra4''}{RA 4''}

\newtheorem{definition}{Definition}

\newtheorem{proposition}{Proposition}

\usepackage{lipsum}
\usepackage{multirow}
\usepackage[noend]{algpseudocode}

\newlength{\algofontsize}
\usepackage[textsize=tiny,textwidth=2cm]{todonotes}

\usepackage{xcolor}
\hypersetup{
	colorlinks,
	linkcolor={red!50!black},
	citecolor={blue!50!black},
	urlcolor={blue!80!black}
}

\newcommand{\prob}{\ensuremath{\mathbf{P}}}

\newcommand{\bE}{\mathbb{E}}

\newcommand{\Var}{\mathrm{Var}}

\newcommand{\Cov}{\mathrm{Cov}}
\newcommand{\1}{\mathbf{1}}

\renewcommand{\mathbf}{\boldsymbol}

\usepackage{tikz}
\usetikzlibrary{shapes,decorations,arrows,calc,arrows.meta,fit,positioning,automata}
\tikzset{
    -Latex,auto,node distance =1 cm and 1 cm,semithick,
    state/.style ={ellipse, draw, minimum width = 0.7 cm,rounded corners},
    point/.style = {circle, draw, inner sep=0.04cm,fill,node contents={}},
    bidirected/.style={Latex-Latex,dashed},
    el/.style = {inner sep=2pt, align=left, sloped}
}
\usepackage{makecell}
\newtheorem{example}{Example}[section]

\usepackage{subcaption}
\usepackage{graphicx}

\usepackage{comment}
\usepackage[suppress]{style-files/color-edits}

\addauthor{ja}{red}
\addauthor{sb}{magenta}
\addauthor{ob}{teal}

\begin{document}
	
	\algrenewcommand\algorithmicrequire{\textbf{Input:}}
	\algrenewcommand\algorithmicensure{\textbf{Output:}}
	
	\title{\vspace{0em} On the Robustness of Second-Price Auctions\\ in Prior-Independent Mechanism Design \vspace{1cm}
	 }
	
	\author{ 
	Jerry Anunrojwong\thanks{Columbia University, Graduate School of Business. Email: {\tt janunrojwong25@gsb.columbia.edu}} \and Santiago Balseiro\thanks{Columbia University, Graduate School of Business. Email: {\tt srb2155@columbia.edu}} \and Omar Besbes\thanks{ Columbia University, Graduate School of Business. Email: {\tt ob2105@columbia.edu}}
	}

\date{\vspace{-1em}}
\maketitle
\begin{abstract}
Classical Bayesian mechanism design relies on the common prior assumption, but the common prior is often not available in practice. We study the design of prior-independent mechanisms that relax this assumption: the seller is selling an indivisible item to $n$ buyers such that the buyers' valuations are drawn from a joint distribution that is unknown to both the buyers and the seller, buyers do not need to form beliefs about competitors, and the seller assumes the distribution is adversarially chosen from a specified class. We measure performance through the worst-case \textit{regret}, or the difference between the expected revenue achievable with perfect knowledge of buyers' valuations and the actual mechanism revenue. 

We study a broad set of classes of valuation distributions that capture a wide spectrum of possible dependencies:  independent and identically distributed (i.i.d.) distributions, mixtures of i.i.d.~distributions, affiliated and exchangeable distributions, exchangeable distributions, and all joint distributions. We derive in quasi closed form the minimax values and the associated optimal mechanism. In particular, we show that the first three classes admit the same  minimax regret value, which is decreasing with the number of competitors, while the last two have the same  minimax regret equal to that of the case $n = 1$. Furthermore, we show that the minimax optimal mechanisms have a simple form across all settings: a  \textit{second-price auction with random reserve prices}, which shows its robustness in prior-independent mechanism design. En route to our results, we also develop a principled methodology to determine the form of the optimal mechanism and worst-case distribution via first-order conditions that should be of independent interest in other minimax problems.    \\

\noindent
\textbf{Keywords}: prior-independent, robust mechanism design, minimax regret, second-price auction with random reserve.
\end{abstract}
\newpage

\setstretch{1.35}

\section{Introduction}




One of the most fundamental questions in all markets is how to optimally sell an item. Selling mechanisms have been  prevalent throughout history, and they have been systematically studied in the literature on \textit{optimal mechanism design} dating back to the seminal work of \cite{Myerson81}. The theory is elegant and has led to an extensive literature. The classical theory, however, has a major conceptual limitation encapsulated through the common prior assumption. It assumes that both the seller and the buyers know the distribution of valuations of buyers and act optimally according to this prior. This leads to mechanisms that depend  on the detailed distributional knowledge and the strategic and informational sophistication of the buyers, which are often not available in practice. This fundamental need to develop mechanisms that are ``more robust'' and less sensitive to modeling assumptions is often referred to as the ``Wilson doctrine'' \citep{Wilson87}. While there have been  an emerging literature on the topic in the last decades, the general derivation of appropriately ``good" mechanisms has been an open question, with few exceptions noted below. In the present paper, we make  progress in understanding ``good" prior-independent mechanisms across a variety of environments. Quite strikingly, the mechanisms we obtain are minimax optimal across a broad variety of classes of distributions typically considered.



In particular, the present paper focuses on the classical mechanism design setting in which a seller wants to sell a good to $n$ buyers. The buyers' valuations are unknown to the seller and are assumed to be drawn from a joint distribution $\mathbf{F}$. In the classical mechanism design literature, $\mathbf{F}$ is assumed to be known to the seller and the buyers. Here, we consider a prior-independent setting: while we still assume that the buyers' valuations are drawn from a distribution, the buyers and the seller do not know this distribution. We will only assume that the seller knows that the distribution belongs to a class $\mathcal{F}$ (we will analyze many such classes). Similarly, we will assume that the buyers do not know the underlying joint distribution. We will therefore focus only on mechanisms that are \textit{dominant strategy incentive compatible} (DSIC): each buyer finds it optimal to report the truth independent of all the other buyers' valuations and strategies. 

To quantify the performance of mechanisms, we will focus on the worst-case \textit{regret}, or the difference between the ideal expected revenue the seller could have collected with knowledge of the buyers' valuations and the expected revenue garnered by the actual mechanism. 
Because the seller does not know the distribution, we assume the seller selects the mechanism that performs well against any element in the class of distribution. In more detail,  we take a minimax approach: the seller selects a mechanism that minimizes the \textit{worst-case} regret over all admissible distributions. We refer to Section~\ref{subsec:intro-model-discuss} for discussions on the choices of the benchmark, objective, as well as our informational assumptions.

Our problem can be understood as a zero-sum game between the seller, who selects a mechanism, and Nature, who best-responds by choosing a distribution. We focus on bounded valuations and in this class, we will also consider a broad set of natural classes of joint distributions of valuations: independent and identically distributed (i.i.d.) distributions, mixtures of i.i.d.~distributions, affiliated and exchangeable distributions, exchangeable distributions, and all joint distributions.  We aim to characterize  the minimax optimal mechanism and its associated performance.

\subsection{Summary of Main Contributions}\label{subsec:summary-main-contributions}

 Our main contributions can be summarized at a high level in  Figure~\ref{fig:dag-results}. The unified framework and techniques we develop yield the minimax optimal mechanisms for regret for any number of buyers $n$ and across the classes of bounded i.i.d.~distributions,\footnote{Note that in the i.i.d.~case, we do not impose shape restrictions on the distribution such as regularity. However, as we will see later, the worst-case distribution is regular, so the minimax optimal mechanism is exactly the same under the class of regular i.i.d.~distributions (cf. Remark after Theorem~\ref{thm:minimax-lambda-regret-main}).} two central classes that capture dependence in valuations (exchangeable and affiliated distributions, and mixtures of i.i.d.~distributions), exchangeable distributions, and all joint distributions. Furthermore, we show that a central mechanism emerges: the same mechanism is minimax optimal across the classes of bounded i.i.d.~distributions, mixtures of i.i.d.~distributions, and affiliated and exchangeable distributions. Our approach also yields optimal mechanisms for the exchangeable distributions and general distributions. A striking feature of our results is that optimal mechanisms have a surprisingly simple form: across all our settings second-price auctions with random reserve prices attain the optimal minimax regret. Our results thus speak to robustness of second-price auctions in prior-independent settings.

\begin{figure}[h!]
\centering
\begin{tikzpicture}
    \node[state,rectangle] (iid) at (0,0) { \makecell[c]{i.i.d. \\distributions}};
    \node[state,rectangle] (aff) at (3,2) {\makecell[c]{exchangeable and\\affiliated distributions}};
    \node[state,rectangle] (mix) at (3,-2) {\makecell[c]{mixture of \\i.i.d. distributions}};
    \node[state,rectangle] (exc) at (7.5,0) {\makecell[c]{exchangeable \\distributions}};
    \node[state,rectangle] (all) at (10.5,0) {\makecell[c]{all \\distributions}};
    \path (iid) edge node[above, el] {$\subseteq$} (aff);
    \path (iid) edge node[above, el] {$\subseteq$} (mix);
    \path (aff) edge node[above, el] {$\subseteq$} (exc);
    \path (mix) edge node[above, el] {$\subseteq$} (exc);
    \path (exc) edge node[above, el] {$\subseteq$} (all);
    \node[draw=blue,thick,dotted,fit=(iid) (aff) (mix), inner sep=0.5cm,label={below:{\textcolor{blue}{\makecell[c]{same minimax optimal \\ mechanism and performance\\(depends on $n$)}}}}] (machine) {};
    \node[draw=blue,thick,dotted,fit=(exc) (all), inner sep=0.5cm,label={below:{\textcolor{blue}{\makecell[c]{same minimax optimal \\ mechanism and performance\\(reduces to $n=1$)}}}}] (machine) {};
\end{tikzpicture}
\caption{Summary of minimax optimal mechanisms and performances derived for different distribution classes.}
\label{fig:dag-results}
\end{figure}
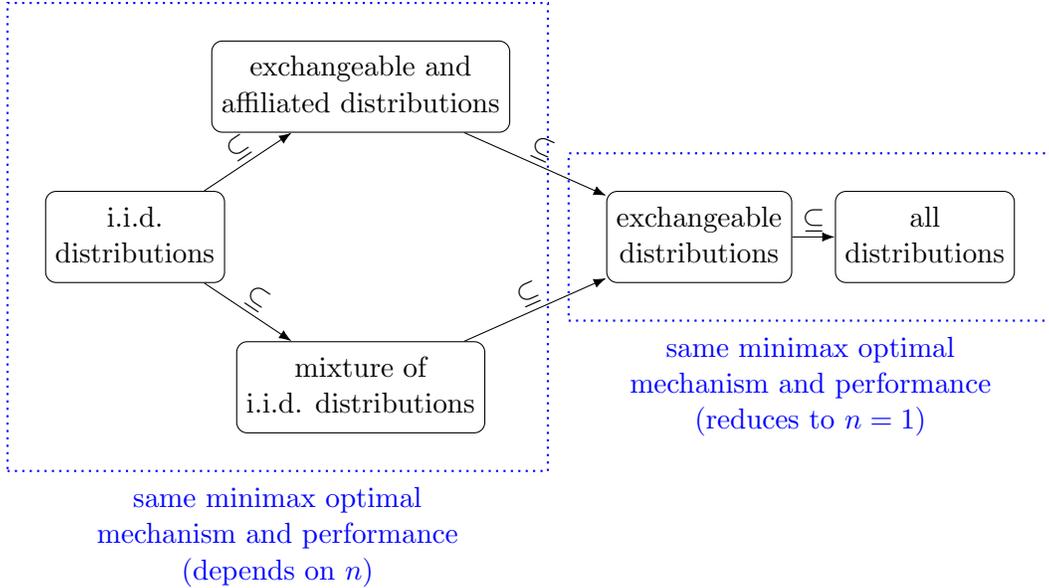


\paragraph{The case of i.i.d.~valuations.} 
We start by  studying the case of i.i.d.~valuations.  Our main innovation is a principled methodology to come up with the candidate optimal mechanism and the worst-case distribution via a saddle point argument. More specifically, assuming that a saddle point exists and that the optimal mechanism is a second-price auction with random reserve (alongside a few regularity conditions), we derive necessary conditions for the structure of a worst-case distribution (cf. Proposition~\ref{prop:opt-F-ode}) and the distribution of reserve prices (Proposition~\ref{prop:opt-phi-ode}). 
We can then establish the optimality of this mechanism, without any additional assumptions, via saddle inequalities. The main result is presented in Theorem \ref{thm:minimax-lambda-regret-main}. In particular, this result yields the exact structure, in quasi-closed form, of the minimax optimal mechanism. 
The distribution of the random reserve price has to be constructed carefully, however, and is determined by a solution to a particular ordinary differential equation. 
 
A corollary of our analysis specialized to the case of a single buyer recovers the results of \cite{BergemannSchlag08}. 

\paragraph{The case of positively dependent valuations.} 
In addition to the i.i.d.~case, we study two central classes that capture the possibility of positive dependence. 
 The first is the notion of affiliation, which was studied under the name \textit{multivariate total positivity of order 2} ($\text{MTP}_2$) in statistics \citep{Karlin68,KarlinRinott80-MTP}. 
 Affiliation was introduced to the mechanism design literature by \cite{MilgromWeber82} and has since become standard in Bayesian mechanism design. To the best of our knowledge, this paper is the first to analyze affiliated distributions in a robust (minimax regret) setting. We establish, quite notably, that the optimal mechanism against i.i.d.~distributions is also minimax optimal against this broader class (Proposition~\ref{prop:minimax-regret-aff}). We establish this result through another saddle point argument, where the key is to control for the impact of allowable dependencies across valuations. We also establish a similar result for the class of mixtures of i.i.d.~distributions, which are used extensively in the statistics literature due to their remarkable modeling power~\citep{hastie2009elements}. 

In other words, as  shown in Figure~\ref{fig:dag-results}, a  second-price auction with carefully constructed distribution of  reserve price is minimax optimal against a variety of classes of distributions. 



\paragraph{General Insights.}
We then use our results to study  a variety of related questions. First, we explore in more detail the structure of optimal mechanisms and the associated performance. In particular, we can quantify, for any number of buyers, the best performance achievable. In Table~\ref{table:minimax-regret-intro}, we provide an example of the minimax regret with valuations in $[0,1]$.

\begin{table}[h!]
\centering

\begin{tabular}{|c || c | c | c | c | c | c | c | c |} 
 \hline
 number of buyers & 1 & 2 & 3 & 4 & 5 & ... & 10 & $\infty$ \\ [0.5ex] 
 \hline
 minimax regret & 0.3679 & 0.3238 & 0.3093 & 0.3021 & 0.2979 & ... & 0.2896 & 0.2815 \\ 
 \hline
\end{tabular}

\caption{Minimax regret as a function of the number of buyers $n$ for valuations in $[0,1]$.}
\label{table:minimax-regret-intro}
\end{table}

In order to obtain a baseline performance and to quantify the value of mechanism design in a prior-independent environment, we also characterize the worst-case performance of two natural benchmark mechanisms: a Vickrey auction, i.e., a second-price auction with no reserve, and also a second-price auction with a deterministic reserve price (cf. Proposition~\ref{prop:sup-regret-spa-r}). Quite notably, the worst-case performance of these mechanisms are similar across the the classes of i.i.d., exchangeable and affiliated, and mixtures of i.i.d.. These results are of independent interest, as they characterize the worst-case performance of commonly used mechanisms. 

We establish that the optimal minimax mechanism yields significant improvements over both mechanisms. In particular, a reserve price is very valuable in prior-independent environments and the use of a randomized reserve price is also critical to minimize the worst-case regret. 

The results also allow us to characterize the impact of increased competition in prior-independent environments. We characterize the minimax regret as the number of agents grows large (cf. Corollary~\ref{cor:lambda-n-infty}).    In particular, we establish that even with large number of buyers, our optimal mechanism (second-price auction with random reserve) still performs better than the standard second-price auction. Our results also highlight that as the number of buyers grow large, the minimax regret  does not shrink to zero as the number of buyers grow large, and the limiting minimax regret can be fully characterized. This apparently counter-intuitive result is driven by the fact that we allow the worst-case distribution to depend on $n$; Nature can select ($n$-dependent) distributions that limit (but not eliminate) the value of increased competition for the seller. Furthermore, this phenomenon is robust to the choice of the benchmark. In fact, when taking  the alternative benchmark associated with  optimal revenues achievable with knowledge of the distribution of values, we can show (Proposition~\ref{prop:minimax-regret-opt-limit}) that for any $n$, the minimax regret also does not shrink to zero as $n$ increases, and is at least $e^{-2} \approx 0.1353$.


\paragraph{The case of arbitrary joint distributions.} 
Finally, we also show that our framework readily yields a minimax optimal mechanism in the case of general exchangeable distributions or general distributions (cf. Proposition~\ref{prop:minimax-regret-arb}). In these two cases, the problem in essence  reduces to the case with a single agent ($n=1$) in terms of achievable performance. Intuitively, with arbitrary joint distributions, nature may simply select only one ``effective" buyer and not allow the seller to capitalize on competition across buyers. Our framework recovers the result for minimax regret against general distributions obtained in \cite{KocyigitKR20-new}. 

We also note that our minimax problems are in general non-convex because of the restriction on the classes of distributions Nature can optimize over.  This is in sharp contrast to other robust mechanism design problems studied in the literature, such as those in \citet{BergemannSchlag08} and \citet{KocyigitRK21-old}, where the minimax problems are bilinear and, thus, efficiently solvable.

More broadly, there has been significant open questions on how to operate in prior-independent environments and on the type of achievable performance in mechanism design in the absence of priors (see, e.g., discussions in \cite{HartlineJohnsen21-blends}).  The present work can be seen as providing a broad spectrum (across distribution classes and number of buyers) of exact/tight characterization of achievability results in mechanism design without priors when the benchmark is the maximal value achievable. The results and techniques developed here may open up the possibility of characterizing minimax performance against alternative benchmarks or with alternative information structures.

\subsection{Related Work}\label{subsec:related-work}

Our work is related to the streams of work on auction design and mechanism design, pioneered by \cite{Vickrey61,Myerson81,RileySamuelson81}. 
We refer the reader to the monographs by \cite{Krishna10-auction-theory-2nd-ed} and \cite{Borgers15-md-book} for a review of classical materials on auction design and mechanism design. These works set a stage for classical \textit{Bayesian} mechanism design for subsequent works, but they share an assumption that the distribution of valuations is common knowledge to both the buyers and the seller. 

The present  work contributes to a stream of work that aims at weakening the knowledge of the seller, in the spirit of Wilson's early call for detail-free mechanisms \citep{Wilson87}. 

In the context of one agent, the problem reduces to that of pricing, and the early works most closely related to our work are \cite{BergemannSchlag08} and \cite{ErenMaglaras10}.
\cite{BergemannSchlag08} considers a minimax regret problem, while \cite{ErenMaglaras10} considers a maximin ratio problem. Both papers consider the \textit{one-bidder} case only. In the original formulations of the two papers, \cite{BergemannSchlag08} considers the first-best benchmark, while \cite{ErenMaglaras10} considers the second-best benchmark. An argument analogous to the proof of our Proposition~\ref{prop:fb-sb-equiv} in Appendix~\ref{app:sec:model-discuss} immediately implies that in the one-bidder case, the first-best and the second-best benchmarks are equivalent.

Both \cite{BergemannSchlag08} and \cite{ErenMaglaras10} specifically consider the random-posted-price mechanism class. In the one-bidder setting, any DSIC mechanism is equivalent to posting a random price, so their settings directly correspond to ours, and \cite{BergemannSchlag08} is an immediate corollary of our results for $n = 1$.

However, there are unique problem features (and associated technical and conceptual complications) that arise in the presence of competing buyers ($n \geq 2$) in our paper that do not arise in the $n = 1$ case. In the latter case, any DSIC mechanism is a random posted price, so the only question involves determining this optimal distribution of posted prices, but in the $n \geq 2$ case, there is no such characterization. In a sense, the class of DSIC mechanisms is ``large'' and ``unstructured.'' In particular, it is not obvious a priori that the optimal mechanism class should be SPA. We consider identifying SPA as a focal mechanism class to consider as one of our main contributions. 

Furthermore, the different distribution classes that we consider in our paper (cf. Figure~\ref{fig:dag-results}) become distinct only for $n \geq 2$; for $n = 1$ they all collapse to the same distribution class. The $n \geq 2$ case is also structurally different from the $n = 1$ case because the $n = 1$ case is \textit{bilinear}, that is, the objective is linear in both the seller's mechanism and Nature's distribution and the problem can be efficiently solved via linear programming duality as in \citep{KocyigitKR20-new,KocyigitRK21-old}. 
In contrast, for $n \geq 2$ the objective is still linear, in the seller's mechanism but is no longer linear in Nature's distribution, and we need to introduce new techniques to solve this problem. 

\cite{KocyigitKR20-new} considers a minimax regret problem with multiple bidders and multiple items. They also consider the first-best benchmark and the class of DSIC mechanisms. Each bidder's utility is assumed to be additive in each item, 
 and the resulting mechanism separates across items. They also let the distributions be arbitrary, and only bounds on the support are known; it is straightforward to adapt our argument from Section 4.3 to imply that arbitrary distributions for any $n$ reduces to $n = 1$. The optimal mechanism exhibited in \cite{KocyigitKR20-new} is the same as our optimal mechanism for $\mathcal{F}_{\textnormal{all}}$ given in Proposition~\ref{prop:minimax-regret-arb}, Section~\ref{subsec:exc-all-dists} in our paper.


The concurrent work of \cite{BachrachTalgamCohen22} considers \textit{maximin revenue} rather than minimax regret, so there is no benchmark. They consider the mechanism with the highest worst-case over the class of i.i.d.~distributions with known upper bound on the support \textit{and known mean}. In the setting of our paper (in which only the upper bound is known) the maximin revenue is trivially zero, achieved when Nature chooses the distribution that is always zero. 
Furthermore, they restrict the mechanism class to second-price-auctions (SPA) with random reserve, and characterize the robustly optimal SPA for every number of bidders $n \ge 1$. 
They prove that a SPA is also optimal in the class of all DSIC mechanisms when $n \le 2$. It remains open to characterize the optimal mechanism for $n \ge 3.$ In contrast, we show that for the minimax regret objective, a SPA is optimal over all DSIC mechanisms for all $n \geq 1$ and over different distribution classes.

In the context of prior-independent mechanism design, very few optimality results are available. In an approximation ratio context, and with the benchmark of optimal revenue with knowledge of the distribution,  initial guarantees against  regular distributions are obtained in \cite{DhangwatnotaiRoughgardenYan15} through a reinterpretation of the classical results of \cite{BulowKlemperer96}, and these are later improved in \cite{fu2015randomization} and \cite{AllouahBesbes20} for two buyers, together with impossibility results. Optimality results are obtained for two buyers with  monotone hazard rate distributions in  \cite{AllouahBesbes20}  and in  \cite{HartlineJohnsenLi20-benchmark-design} for regular distributions. We also refer the reader to \cite{Hartline-approx-md-book} for a survey.

Our work also relates to the design of optimal pricing and mechanisms with limited data. Some works assume access to samples drawn from the distribution
\citep{ColeRoughgarden14,DhangwatnotaiRoughgardenYan15,making-the-most-of-your-samples,bubeck-et-al-multi-scale-online-learning,AllouahBahamouBesbes22-pricing-with-samples,FengHartlineLi21-revelation-gap-pricing-samples,FuHHK21}. Other works assume knowledge of summary statistics of distributions; see, e.g., \cite{AzarDMW13,AllouahBahamouBesbes-pricing-single-point,Suzdaltsev20-dr-auction,Suzdaltsev-dr-pricing}.

Here, we focus on dominant strategy incentive compatible mechanisms. This is also motivated by the fact that buyers often have limited aggregate information about the market. We refer the reader to \cite{ChungEly07} for a formal justification for this assumption.
  Other works make different assumptions about the buyers. For example, \cite{BoseOzdenorenPape06,ChiesaMicaliZhu15-knightian-vickrey} assume that the buyers have ambiguity sets about the valuations or distributions of other buyers.  


Robustness has  received significant attention in the mechanism design and contracting literatures and we refer the reader to the recent survey  \citep{Carroll19-robust-survey}. We highlight here that a different form of robustness is often also studied: robustness to the information structure such as higher order beliefs and type spaces; see, e.g., \cite{BergemannMorris05-robust-md-ecta}. The concurrent work of \cite{Che22} studies optimal auction design with mean constraints that is robust in this sense.  In particular, he considers the class of \textit{arbitrary} (not necessarily i.i.d.) distributions with bounded support and known marginal means, and the objective is maximin revenue (no benchmark).  Most importantly, the seller is allowed to use \textit{any} mechanism, that is not necessarily DSIC, and the solution concept is Bayes-Nash equilibrium rather than dominant strategy equilibrium. The calculation of Bayes-Nash equilibrium involves not just the valuation distribution itself, but also the \textit{information structure}, i.e., who knows what, and who knows about who knows what, and so on (higher order beliefs). \cite{Che22} assumes that Nature, apart from adversarially selecting a distribution from a distribution class, also adversarially selects an information structure as well. He first restricts the mechanism class to SPA and characterizes an optimal distribution of reserve prices. He then shows that the SPA is optimal among a broader class of mechanisms he calls ``competitive'' but it remains open to show whether a SPA is optimal among all mechanisms. 



Lastly, there are also works on optimal robust mechanisms  for selling multiple goods under various objectives and knowledge assumptions; see, e.g.,  \cite{Carroll17-robust-screening,CheZhong21}.

\section{Problem Formulation}\label{sec:problem-formulation}



The seller wants to sell an indivisible object to one of $n$ buyers. The $n$ buyers have valuations drawn from a joint cumulative distribution $\mathbf{F}$.\footnote{Throughout, we use  boldface to denote joint cumulative distributions or vectors.}  We assume that the seller does not know the distribution of valuations of buyers $\mathbf{F}$, and only knows the upper bound of the valuation of each agent, normalized to $1$. That is, the seller only knows that the  support of the buyers' valuations belongs to $[0,1]^{n}$. The goal of the seller is to design a mechanism that minimizes her worst-case \textit{regret} given the limited information about the underlying distribution of valuations of buyers, where the regret is defined as the difference between  the revenue  the seller could achieve with knowledge of the valuations and the ones actually garnered. With knowledge of the valuations, the maximum revenue achievable is given by the maximum valuation of all the buyers $\max(\mathbf{v}) = \max(v_1,\dots,v_n)$. In other words, the \textit{difference} between the benchmark and the revenue quantifies the \textit{value} of knowing the agents' valuations. 

\paragraph{Seller's Problem.} We model our problem as a game between the seller and Nature, in which the seller first selects a selling mechanism and then Nature may counter such a mechanism with any distribution of buyers’ valuations with the admissible support. Buyers' valuations are then drawn from the distribution chosen by Nature and they participate in the seller's mechanism.

A selling mechanism $m = (\mathbf{x}, \mathbf{p})$ is characterized by an allocation rule $\mathbf{x}$ and a payment rule $\mathbf{p}$, where $\mathbf{x}: \mathbb{R}^n \to [0,1]^n$ and $\mathbf{p}: \mathbb{R}^n \to \mathbb{R}$. Given buyers' valuations $\mathbf{v} \in [0,1]^n$, $x_i(\mathbf{v})$ gives the probability that the item is allocated to buyer $i$, and $p_i(\mathbf{v})$ his expected payment to the seller.
We will restrict attention to dominant strategy incentive compatible (DSIC) direct mechanisms. For such mechanisms,
buyers need not make any assumptions about the underlying distribution of valuations and will find it optimal to report their true valuation, independently of the realization of valuations of the other buyers.
More formally, we require that the mechanism $m = (\mathbf{x}, \mathbf{p})$ satisfies the following constraints:
\begin{align*}
    v_i x_i(v_i, \mathbf{v}_{-i}) - p_i(v_i,\mathbf{v}_{-i}) &\geq 0, \quad \forall i, v_i, \mathbf{v}_{-i} &\text{(IR)}\\
    v_i x_i(v_i, \mathbf{v}_{-i}) - p_i(v_i,\mathbf{v}_{-i}) &\geq v_i x_i(\hat{v}_i, \mathbf{v}_{-i}) - p_i(\hat{v}_i,\mathbf{v}_{-i}) \quad \forall i,v_i,\mathbf{v}_{-i},\hat{v}_i &\text{(IC)} \\
    \sum_{i=1}^{n} x_i(v_i,\mathbf{v}_{-i}) &\leq 1 \quad \forall \mathbf{v} &\text{(AC)}
\end{align*}
The ex-post individual rationality constraint (IR) states that each buyer $i$ is willing to participate because his payoff is at least his outside option, normalized to zero. The ex-post incentive compatibility constraint (IC) states that each buyer $i$ always finds it optimal to report his true valuation, regardless of the valuation of other buyers.   Lastly, the allocation constraint (AC) states that the seller can allocate at most one item. Note that we allow the seller's mechanism to be randomized. 



Throughout this paper, we will focus on the class of all mechanisms satisfying these constraints, which we will call DSIC mechanisms: 
\begin{align}\label{mech-ir-ic-ac}
    \mathcal{M} = \left\{ (\mathbf{x},\mathbf{p}): \text{(IR), (IC), (AC)} \right\}.
\end{align}


Informally, the seller seeks to maximize the expected revenue $\bE_{\mathbf{v} \sim \mathbf{F} } \left[\sum_{i=1}^{n} p_i(\mathbf{v}) \right]$ relative to the benchmark associated with the revenues that could be collected when the valuations of the buyers are known $\bE_{\mathbf{v} \sim \mathbf{F}} \left[ \max(\mathbf{v}) \right]$. We consider the regret objective defined by
\begin{align}
    \text{Regret}(m,\mathbf{F}) &= \bE_{\mathbf{v} \sim \mathbf{F}} \left[ \max(\mathbf{v}) - \sum_{i=1}^{n} p_i(\mathbf{v}) \right] \label{eq:regret} 
\end{align}

After the seller chooses a mechanism $m$, Nature then chooses a distribution $\mathbf{F}$ from a given class of distributions $\mathcal{F}$ such that the valuation of the $n$ agents $\mathbf{v} \in \mathbb{R}_{+}^n$ are drawn from $\mathbf{F}$. The seller aims to select the mechanism $m$ to minimize the \textit{worst-case} regret. Our goal, therefore, is to characterize the minimax regret
\begin{align}
  \mathcal{R}_{n}^*(\mathcal{F}) :=  \inf_{m \in \mathcal{M}} \sup_{\mathbf{F} \in \mathcal{F}} \text{Regret}(m, \mathbf{F}).
\end{align}

The problem is specified by the choice of the class of mechanisms $\mathcal{M}$, and the choice of the class of distributions $\mathcal{F}$. As stated before, we will  consider the class of all DSIC direct mechanisms $\mathcal{M}$ satisfying (IR), (IC), (AC) as defined in (\ref{mech-ir-ic-ac}).




\paragraph{Admissible distributions.} Lastly, we consider the choice of the class of admissible distributions  $\mathcal{F}$. This class can be seen as capturing the ``power'' of Nature: the larger the class, the more powerful/adversarial Nature becomes. To date, when imposing no shape constraints in the distributions, the only available results for this class of problems are for the regret objective and $n=1$ \citep{BergemannSchlag08}, and for arbitrary $n$ when the class $\mathcal{F}$ is the class of arbitrary joint distributions \citep{KocyigitKR20-new}. Interestingly, the minimax regret in the latter case is equal to that when $n=1$. In other words, without further restrictions on its power, Nature  can simply eliminate competition and more buyers do not improve minimax regret against arbitrary distributions. 

In the present work, we  consider a broad set of natural candidates for $\mathcal{F}$ that represent a spectrum of possible levels of power for Nature. In particular, we will consider typical classes considered in the mechanism design literature when the distribution is assumed to be known, ranging from the independently and identically distributed (i.i.d.) valuations case, to various notions of common positive dependence structures (affiliated and exchangeable valuations and mixtures of i.i.d.~valuations), to general exchangeable distributions and general distributions.

\subsection{Problem Formulation Discussion}\label{subsec:intro-model-discuss}

\obedit{Having laid out the problem formulation, we study, we now discuss different features of the formulation}. \obdelete{We give a more extensive discussion in Appendix~\ref{app:sec:model-discuss}.}

\paragraph{Benchmark.} We consider the \textit{first-best} benchmark, the maximum revenue achievable when the \textit{valuations} are known. Another reasonable choice of the benchmark is the \textit{second-best}, the maximum revenue achievable when the \textit{distribution} is known. Both the first-best benchmark and the second-best benchmark are extensively used in the literature, and there is not necessarily a single universal benchmark. The first-best benchmark has been used in the economics, operations, and computer science literatures such as \citep{BergemannSchlag08,CaldenteyLiuLobel17,KocyigitKR20-new,robust-monopoly-regulation,kleinberg-yuan}. The first-best benchmark is also reminiscent of the \textit{offline optimum} benchmark, which is extensively used in the analysis of algorithms \citep{BorodinElYaniv}. An advantage of the first-best benchmark is that it is easily computable and can be evaluated counterfactually (using the reported values when the mechanism is DSIC).

We can show (cf. Proposition~\ref{prop:fb-sb-equiv}) that when the distribution class consists of all distributions, minimax regret against the first-best benchmark and the second-best benchmark are equal, so the choice of the benchmark is immaterial in that case. 

When the distribution class is i.i.d., the choice of the benchmark matters. Extending our work to the characterization of the minimax regret against the second-best benchmark when the distribution class is i.i.d.~is an important direction for future work. In general, it is not clear a priori if it is possible to solve for the optimal minimax mechanism. In this paper, we were able to show that it is possible to do so against the first-best benchmark through a pure saddle point approach. Characterizing a minimax mechanism against the second-best benchmark would require a different approach because a pure saddle point cannot exist;  
we provide a more extensive discussion in Appendix~\ref{app:subsec:model-discuss-benchmark}. Instead, one would need to identify a \textit{mixed} saddle point, which consists of an optimal mechanism and the corresponding \textit{distribution over distributions}, which is, at least at this stage, a much harder object to deal with.

\paragraph{Upper Bound Information.}  \obedit{Throughout the main text, we assume that the values have support included in $[0,1]$.  This is purely a normalization for expositional convenience. If the support is $[0,b]$ for some known upper bound $b$, then the mechanism, distribution, and regret are all scaled by $b$ (we formalize this in Proposition~\ref{prop:upper-bound-b}).}

The motivation for the knowledge of an upper bound on the support  is that we operate in a world with minimal or no data, and significant ``Knightian'' uncertainty. In particular, we think of the upper bound as some ``minimal amount of information'' that can be collected even in the absence of data. Upper bounds can be obtained by consulting experts, using domain knowledge, or common sense. Some examples might be launching a new product, or auctioning rarely traded goods (fine art, collectibles, jewelry). In these contexts, it might be  easier or more intuitive for experts to come up with a range of values than distributional parameters such as the mean of the distribution of values, derived quantities such as the optimal monopoly price, or the exact shape of the distribution. 

We emphasize that the upper bound simply defines an absolute upper bound on the feasible region of the distributions and need not be the exact actual support of the distribution of values. For example, the seller can err on the side of caution and choose a conservative upper bound. Even when the upper bound is ``misspecified,'' the performance degrades gracefully with the error; we provide an extensive discussion and analysis in Appendix~\ref{app:subsec:model-discuss-upper-bound}.

\paragraph{DSIC Mechanism Class.} 
Our work focuses on dominant strategy as a notion of robustness because this notion of robustness protects against uncertainty in both the \textit{valuation distribution} and \textit{bidder behavior}. Under dominant strategy incentive compatibility, each bidder is incentivized to report her value truthfully regardless of all other bidder's behavior. This solution concept is very compelling because the optimal bidder strategy is simple and, therefore, predictable to the seller and bidders alike, and it levels the playing field between bidders with potentially different knowledge and ability to strategize~\citep{level-playing-field}. In contrast, if we further extend our mechanism class to non-DSIC mechanisms, some bidders may behave strategically, and how they behave will depend intricately on the amount of their knowledge and sophistication, which are not necessarily known to the seller. The choice of DSIC mechanisms allows us to avoid having to make strong assumptions about bidder behavior. 

\section{Analysis of Minimax Regret for i.i.d.~
Valuations}\label{sec:minimax-regret-main}











In this section, we focus on the class $\mathcal{F}_{\text{iid}}$ of i.i.d.~valuations with support in $[0,1]^n$, which we can define formally as follows. 

\begin{definition}
The class  $\mathcal{F}_{\text{iid}}$ consists of all distributions 
such that there exists a distribution $F$ with support on $[0,1]$, referred to as the marginal, such that $\mathbf{F}(\mathbf{v}) = \prod_{i=1}^{n} F(v_i)$ for every $\mathbf{v} \in [0,1]^n$.
\end{definition}

Before introducing the main result, we recall the definition of a  second-price auction (SPA) with reserve as such mechanisms will appear in the result.  Let $\mathbf{v} = (v_1,\dots,v_n) \in \mathbb{R}_+^n$ be the valuations of the $n$ agents and $v^{(1)} \geq v^{(2)} \geq \cdots \geq v^{(n)}$ be the order statistics.  A second-price auction with reserve $p$ allocates the item to the highest bid (breaking ties arbitrarily) and charges payment $\max(v^{(2)},p)$ if the highest bid is at least $p$, and otherwise, does not allocate or charge anything. When $p$ is drawn from a distribution, the mechanism is a second-price auction with random reserve price. It is straightforward that such mechanisms are DSIC.

We are now ready to give the statement of our main result that fully characterizes the optimal mechanism and optimal performance of the minimax regret problem.

\begin{theorem}[Optimal mechanism and optimal performance]\label{thm:minimax-lambda-regret-main}
The minimax regret admits as an optimal mechanism a second-price auction with random reserve price with cumulative distribution $\Phi_{n}^*$ on $[r_{n}^*,1]$ given by
\begin{align*}
    \Phi_{n}^*(v) =  \left( \frac{v}{v-r_{n}^*} \right)^{n-1} \log \left( \frac{v}{r_{n}^*} \right) - \sum_{k=1}^{n-1} \frac{1}{k} \left( \frac{v}{v-r_{n}^*} \right)^{n-1-k},
\end{align*}
where $r_{n}^* \in (0,1/n)$ is the unique solution to 
\begin{align*}
    (1-r^*)^{n-1} + \log(r^*) + \sum_{k=1}^{n-1} \frac{(1-r^*)^k}{k} = 0.
\end{align*}
Furthermore, the minimax regret is given by 
\begin{align*}
 \mathcal{R}_{n}^*(\mathcal{F}_{\text{iid}})  =  (1-r_{n}^*)^{n-1} -  \int_{v=r_{n}^*}^{v=1} \left( 1 - \frac{r_{n}^*}{v} \right)^{n-1} dv.
\end{align*}
\end{theorem}
This result provides, in quasi-closed form, a minimax regret optimal mechanism for any number of buyers when the valuations are i.i.d., and the corresponding optimal performance. Quite remarkably, the minimax optimal mechanism admits a simple structural form: a second-price auction with reserve, but the minimax optimal reserve price is random. We present and discuss these in detail, together with more intuition on the structure of the optimal mechanism, in Section \ref{sec:char-regret-ratio}.


\paragraph{Remark.} As we will see in the proof, the worst-case distribution against the optimal mechanism  above happens to have non-decreasing virtual values on the interior of its support, with a mass at the maximum of its support. In turn, the result above implies that we have also identified the minimax regret when restricting attention to regular distributions with finite support. 

\paragraph{Remark on the case $n=1$.} An important corollary is the $n = 1$ agent case, where we get the results of \cite{BergemannSchlag08} associated with support $[0,1]$ as a special case. 
\begin{corollary}\label{cor:lambda-n-1}
Suppose $n=1$. We have $r_{1}^* = 1/e$ and $\Phi_{1}^*(v) =  \log(v) + 1$ for $v \in \left[1/e, 1 \right] $. The minimax regret is $1/e$.
\end{corollary}




\subsection{Key Ideas for the Proof of Theorem \ref{thm:minimax-lambda-regret-main}} 

Below, we detail the key ideas underlying the derivation of a candidate minimax regret optimal mechanism. The proof is presented in Section \ref{subsec:proof-main-theorem}. 


The proof relies on a saddle point argument. Formally, a saddle point is defined as follows. 
\begin{definition}[Saddle Point]
$(m^*,\mathbf{F}^*)$ is the saddle point of $R_{n}(m,\mathbf{F})$ if 
\begin{align*}
    R_{n}(m^*,\mathbf{F}) \leq R_{n}(m^*,\mathbf{F}^*) \leq R_{n}(m,\mathbf{F}^*) \qquad \mbox{for all $m \in \mathcal{M},\mathbf{F} \in \mathcal{F}$}
\end{align*}
\end{definition}
Note that if $(m^*,\mathbf{F}^*)$ is a saddle point then $\inf_{m} \sup_{\mathbf{F}} R_{n}(m,\mathbf{F}) \leq \sup_{\mathbf{F}} R_{n}(m^*,\mathbf{F}) \leq R_{n}(m^*,\mathbf{F}^*) \leq \inf_{m} R_{n}(m,\mathbf{F}^*) \leq \inf_{m} \sup_{\mathbf{F}} R_{n}(m,\mathbf{F})$, which allows us to conclude that $R_{n}(m^*,\mathbf{F}^*) = \inf_{m} \sup_{\mathbf{F}} R_{n}(m,\mathbf{F})$. So to solve the minimax problem, it is sufficient to exhibit a saddle point and prove the corresponding saddle inequalities. Existence of a saddle point implies that the simultaneous-move game between the seller and Nature in which the seller chooses a mechanism and Nature a distribution admits a Nash equilibrium.

Our approach relies on two main parts: i.) a principled method to pin down (under additional assumptions) what the structure of $m^*$ and $\mathbf{F}^*$ should be, and then ii.) verify formally that the pair $(m^*,\mathbf{F}^*)$ derived earlier is indeed a saddle point. 
 
 Our approach for i.) relies on two main steps. We first conjecture that $m^*$ is a second-price auction (SPA) with random reserve drawn from a distribution $\Phi^*$ which admits a density. Then, we provide an explicit formula for the  regret  $R_{n}(\Phi,\mathbf{F})$. In particular, we show that this regret admits two possible representations:
\begin{itemize}
    \item An expression that involves only the marginal $F(v)$ and not its derivative. This expression will appear in \eqref{eq:regret-F} in Proposition \ref{prop:regret-expressions-general}.
    \item  An expression that only depends on $\Phi$ and not its derivative, and is explicitly linear in $\Phi$. This expression will appear in \eqref{eq:regret-linear-phi} in Proposition \ref{prop:regret-expressions-general}. 
\end{itemize}


Based on these two alternative representations,  we heuristically derive a candidate saddle point using an ``inverse'' optimization approach. 

To do so,  first, we use the seller's first-order conditions for an optimal distribution of reserve prices $\Phi^*$ to derive a worst-case distribution $\mathbf{F}^*$. In particular, because the regret can be seen as being linear in $\Phi$ (cf. \eqref{eq:regret-linear-phi}), if we plug-in the (unknown) worst-case distribution $\mathbf{F}^*$, the first-order conditions for the seller (under some conditions) dictate that the coefficient of each $\Phi^*(v)$ should be zero whenever the optimal distribution of reserve prices $\Phi^*$ is interior. Because this coefficient only depends on the distribution, this gives an equation that $F^*$ should satisfy, pinning the value of a candidate $\mathbf{F}^*$. 

Second, we use Nature's optimality conditions for the worst-case distribution $\mathbf{F}^*$ to pin down the candidate optimal distribution of reserves $\Phi^*$. The crucial ansatz we make is that both distributions have the same support $[r^*,1]$. We require that the distribution $\mathbf{F}^*$ we derived is actually worst-case optimal for Nature. Because the regret expression in \eqref{eq:regret-F} only involves the marginal distribution, fixing $\Phi^*$, we show that pointwise optimization leads  to simple optimality conditions for the worst-case distribution $F^*$. By requiring that the worst-case distribution is $F^*$, we obtain an ODE equation that pins down the  distribution of reserve prices $\Phi^*$. Lastly, we show that a unique $r^*$ is possible. In Appendix~\ref{app:sec:derive-candidates}, we give more details on our derivation of candidate mechanisms and distributions in the saddle point.

Finally, once we have the candidate $m^* = \text{SPA}(\Phi^*)$ and $\mathbf{F}^*$, we formally verify the saddle inequalities in Section~\ref{subsec:proof-main-theorem}. Importantly, the saddle verification does not require any of the additional assumptions that we made to obtain the candidate saddle point and so the result yields the actual minimax optimal mechanism (in the space of all DSIC mechanisms) and its performance.

\subsection{Regret of a Second-Price Auction with Random Reserve}\label{subsec:regret-spa-random-reserve}


\obcomment{Why don't we combine 3.1 and 3.2 with new changes?}

Consider a  second-price auction with random reserve price distribution $\Phi$. In such an auction,  first, all $n$ agents submit their bids $\mathbf{v} \in \mathbb{R}_+^n$, then we draw the random reserve price $p$ from the distribution $\Phi$ and finally the seller runs the second-price auction with reserve $p$ to determine the allocation and payment. We denote such a mechanism as $\text{SPA}(\Phi)$. The expected revenue of $\text{SPA}(\Phi)$ is $\bE_{p \sim \Phi} \left[ \max(v^{(2)},p) \1(v^{(1)} \geq p) \right]$. As noted earlier,  for any $\Phi$, $\text{SPA}(\Phi) \in \mathcal{M}$: the second-price auction with random reserve mechanism is DSIC. 

We will now derive the regret of $\text{SPA}(\Phi)$ against the distribution $\mathbf{F}$, denoted $R_{n}(\Phi,\mathbf{F})$. We will further assume that $\Phi$ has a density and is supported on $[r^*,1]$ because that is all we need for the proof. We will first derive the (Regret-$\mathbf{F}$) expression that is valid for any joint distribution $\mathbf{F}$. Then we will specialize to the case when $\mathbf{F}$ is i.i.d., and we get the \eqref{eq:regret-F} expression and an alternative \eqref{eq:regret-linear-phi} expression.

\begin{proposition}\label{prop:regret-expressions-general}

Assume that $\Phi$ has a density and is supported on $[r^*,1]$. Then the regret of $\textnormal{SPA}(\Phi)$ under distribution $F$, denoted $R_{n}(\Phi,F)$, is given by
\begin{align*}
    R_{n}(\Phi,F) &= r^* -  \int_{v=0}^{v=r^*} \mathbf{F}_n^{(1)}(v) dv \nonumber\\
    &+ \int_{v=r^*}^{v=1} (1 -v \Phi'(v) - \Phi(v) ) (1-\mathbf{F}_n^{(1)}(v) ) dv +  \Phi(v) ( \mathbf{F}_n^{(2)}(v) - \mathbf{F}_n^{(1)}(v) ) dv \tag{Regret-$\mathbf{F}$}
\end{align*}
where $\mathbf{F}_n^{(1)}(p) = \Pr_{\mathbf{F}} \left( v^{(1)} \leq p \right)$ and $\mathbf{F}_n^{(2)}(p) = \Pr_{\mathbf{F}} \left( v^{(2)} \leq p \right)$ are the CDFs of the first- and second-highest order statistics, respectively, and the subscript $n$ is to explicitly note that we compute the probability over $n$ agents. 

If we assume that $\mathbf{F}$ is i.i.d.~and write $F(\cdot)$ for each agent's marginal CDF, then
\begin{align} 
    R_{n}(\Phi,\mathbf{F}) &= r^*   - \int_{v=0}^{v=r^*} F(v)^n dv \nonumber\\
    & + \int_{v=r^*}^{v=1} (1 -v \Phi'(v) - \Phi(v) ) (1-F(v)^n ) +  \Phi(v) n F(v)^{n-1}(1-F(v)) dv \tag{Regret-$F$} \label{eq:regret-F}
\end{align}

Furthermore, if we assume that $F$ is supported on $[r^*,1]$ and $F'(v)$ exists on $(r^*,1)$, then
\begin{align*} 
    R_{n}(\Phi, \mathbf{F}) &= r^* -  \int_{v=0}^{v=r^*} F(v)^n dv \nonumber \\
    & + \int_{v=r^*}^{v=1} (1-F(v)^n)  + nF(v)^{n-1}(1-F(v)-vF'(v))\Phi(v) dv \tag{\text{Regret}-$\Phi$} \label{eq:regret-linear-phi}
\end{align*}

\end{proposition}

\subsection{Proof of Theorem \ref{thm:minimax-lambda-regret-main} }\label{subsec:proof-main-theorem}

We are now ready to prove our main result, Theorem~\ref{thm:minimax-lambda-regret-main}.


\begin{proof}[Proof of Theorem~\ref{thm:minimax-lambda-regret-main}]

    We define $\mathbf{F}_{n}^* = (F_{n}^*)^n$ to be the  i.i.d.~distribution over $n$ valuations with marginal  $F_{n}^*$, and $F_{n}^*$ is the isorevenue distribution starting at $r_{n}^*$ as given by $F_{n}^*(v) = 1-r_{n}^*/v$ for $v \in [r_{n}^*,1)$ and $F_{n}^*(1) = 1$. We also define the optimal mechanism $m_{n}^* = \text{SPA}(\Phi_{n}^*)$. 

\paragraph{Step 1 (Nature's optimality).} We will first prove that $R_{n}(m_{n}^*,\mathbf{F}) \leq R_{n}(m_{n}^*, \mathbf{F}_{n}^*)$, or equivalently, $R_{n}(\Phi_{n}^*,\mathbf{F}) \leq R_{n}(\Phi_{n}^*,\mathbf{F}_{n}^*)$. Note that $\Phi_{n}^*$ has a density, and is supported on $[r_{n}^*,1]$, so $R(\Phi_{n}^*,\mathbf{F})$ is given by the (Regret-$F$) expression in Proposition~\ref{prop:regret-expressions-general}. Because $\int_{v=0}^{v=r_{n}^*} F(v)^n \geq 0$, we have
\begin{align*}
    R_{n}(\Phi_{n}^*,\mathbf{F}) &\leq  r_{n}^*  +  \int_{v=r_{n}^*}^{v=1} (1 -\Phi_{n}^*(v) - v (\Phi_{n}^*)'(v))(1-F(v)^n) + n F(v)^{n-1} (1-F(v)) \Phi_{n}^*(v) dv\\
    &= r_{n}^*  +  \int_{v=r_{n}^*}^{v=1} \left\{  \frac{n r_{n}^* - v}{v - r_{n}^*} + n F(v)^{n-1} - \frac{(n-1)v}{v-r_{n}^*} F(v)^{n} \right\} \Phi_{n}^*(v) dv \\
    &\leq r_{n}^*  +  \int_{v=r_{n}^*}^{v=1} \sup_{z \in [0,1]} \left\{ \frac{n r_{n}^* - v}{v - r_{n}^*} + n z^{n-1} - \frac{(n-1)v}{v-r_{n}^*} z^{n} \right\} \Phi_{n}^*(v) dv
\end{align*}
where the equality in the second line follows from \obedit{from the fact that satisfies the ordinary differential equation \eqref{eq:ODE-Phi}: $(\Phi_{n}^*)'(v) + \frac{(n-1) r_{n}^*}{v(v-r_{n}^*)} \Phi_{n}^*(v) = \frac{1}{v}$; this fact is established in Proposition \ref{prop:ode-phi-solution}. Finally,  the last inequality is a pointwise bound. }
\obcomment{SHOULD WE STATE THIS PROPOSITION HERE? }

The derivative of the expression in inside the supremum with respect to $z$ is $n(n-1) z^{n-2} (v-r_{n}^*-vz )/(v-r_{n}^*)$ which is less than or equal to zero for $z \leq 1-r_{n}^*/v$ and is greater than or equal to zero for $z \geq 1-r_{n}^*/v$. Therefore, it is maximized at $z = F_{n}^*(v) = 1 - r_{n}^*/v$. As a result,
\begin{align*}
    R_{n}(\Phi_{n}^*,\mathbf{F})
    &\leq r_{n}^*  +  \int_{v=r_{n}^*}^{v=1}  \left\{ \frac{n r_{n}^* - v}{v - r_{n}^*} + n F_{n}^*(v)^{n-1} - \frac{(n-1)v}{v-r_{n}^*} F_{n}^*(v)^{n} \right\} \Phi_{n}^*(v) dv = R_{n}(\Phi_{n}^*,\mathbf{F}_{n}^*).
\end{align*}

\paragraph{Step 2 (Seller's optimality).}    Now we want to show that $R_{n}(m_{n}^*,\mathbf{F}_{n}^*) \leq R_{n}(m,\mathbf{F}_{n}^*)$. That is, if we fix the distribution to be $\mathbf{F}_{n}^*$, then $m^*$ achieves the lowest regret among all mechanisms. We note that $R_{n}(m,\mathbf{F}_{n}^*) = \bE_{\mathbf{v} \in \mathbf{F}_{n}^*} \left[  \max(\mathbf{v}) \right] - \bE_{\mathbf{v} \in \mathbf{F}_{n}^*} \left[ \sum_{i=1}^{n} p_i(\mathbf{v}) \right]$. The first term does not depend on $m$ and the second term is the expected revenue under $\mathbf{F}_{n}^*$. Therefore, we want to show that $m$ maximizes the expected revenue under $\mathbf{F}_{n}^*$.   Our calculation is analogous to \cite{Myerson81} but we cannot use the standard results directly because $\mathbf{F}_{n}^*$ is a mixture of a point mass and an absolutely continuous distribution with density, so we will work out the expected revenue expression from first principles. The details are standard (because it is a Bayesian mechanism design problem) and are provided in Appendix~\ref{app:sec:minimax-regret-main} for completeness.
\end{proof}

\section{Minimax Regret under Various Families with Dependent Valuations}\label{sec:extend-other-F}

So far, we have characterized the minimax regret $ \inf_{m \in \mathcal{M}} \sup_{\mathbf{F} \in \mathcal{F}} R_{n}(m,\mathbf{F})$ only for the class $\mathcal{F}_{\text{iid}}$ of $n$ i.i.d.~agents. 
In this section, we extend our framework by considering other choices of the class of distribution $\mathcal{F}$ that capture different notions of dependence between agents' valuations.

\subsection{Exchangeable and Affiliated Distributions}\label{subsec:exc-aff-dists}

In this section, we analyze the class of exchangeable and affiliated distributions (referred to as multivariate total positivity of order 2, or $\text{MTP}_2$, in statistics). The notion of affiliation is a notion of positive dependence that was introduced to the mechanism design literature by \cite{MilgromWeber82} and has since become standard.  We restrict attention to distributions that are symmetric as, otherwise, Nature can eliminate competition by setting the valuations of all but one bidder to zero. 


\begin{definition}
A distribution $\mathbf{F}$ over $n$ real valuations is \emph{exchangeable} if it has the same joint distribution under any ordering of those valuations. Formally, a vector of random variables $\mathbf{v} = (v_1,\dots,v_n)$ is exchangeable if for any permutation $\sigma: \{1,2,\dots,n\} \to \{1,2,\dots,n\}$, the random variables $(v_{\sigma(1)},v_{\sigma(2)},\dots,v_{\sigma(n)})$ and $(v_1,\dots,v_n)$ have the same joint distributions. A joint distribution $\mathbf{F}$ is exchangeable if the corresponding random variable is exchangeable.
\end{definition} 

To give a formal definition of affiliation, we need a few auxillary definitions.\footnote{We follow the treatment in the appendix of \cite{MilgromWeber82}.} First, a subset $A$ of $\mathbb{R}^n$ is called \textit{increasing} if its indicator function $\bold{1}_{A}$ is nondecreasing. Second, a subset $S$ of $\mathbb{R}^n$ is a \textit{sublattice} if its indicator function $\mathbf{1}_S$ is affiliated, i.e., for any  $\mathbf{v},\mathbf{v}' \in S$ we have $\mathbf{v} \land \mathbf{v}' \in S$ and $\mathbf{v} \lor \mathbf{v}' \in S$ where $\mathbf{v} \land \mathbf{v}' = (\min(v_1,v_1'),\dots,\min(v_n,v_n'))$ and $\mathbf{v} \lor \mathbf{v}' = (\max(v_1,v_1'),\dots,\max(v_n,v_n'))$ denote the component-wise minimum and maximum of $\mathbf{v}$ and $\mathbf{v}'$, respectively. For a distribution $\mathbf{F}$ and sets $A$ and $S$, we write $\Pr_{\mathbf{F}}(A|S)$ as the probability that a random variable drawn from the conditional distribution $\mathbf{F}$ restricted to $S$ is in $A$.


We can now give the definition of affiliation.

\begin{definition}
A distribution $\mathbf{F}$ is \textit{affiliated} if for all increasing sets $A$ and $B$ and every sublattice $S$, $\Pr_{\mathbf{F}}(A \cap B | S) \geq \Pr_{\mathbf{F}}(A|S) \Pr_{\mathbf{F}}(B|S)$.
\end{definition}

The following property of affiliation is well-known (see e.g. Theorem 24 of \cite{MilgromWeber82}) and we will use it throughout. Assume that the distribution $\mathbf{F}$ has a density $f$. $\mathbf{F}$ is affiliated if and only if $f$ satisfies the affiliation inequality $f(\mathbf{v}) f(\mathbf{v}') \leq f(\mathbf{v} \land \mathbf{v}') f(\mathbf{v} \lor \mathbf{v}')$ for almost every $\mathbf{v}, \mathbf{v}'$. If $\mathbf{F}$ is a discrete distribution, the same relationship holds when we interpret $f$ as the probability mass function.



We denote by $\mathcal{F}_{\text{aff}}$ the set of all exchangeable and affiliated distributions with support $[0,1]^n$. Clearly $\mathcal{F}_{\text{iid}} \subset \mathcal{F}_{\text{aff}}$ and hence $\mathcal{R}_{n}^*(\mathcal{F}_{\text{aff}})  \ge \mathcal{R}_{n}^*(\mathcal{F}_{\text{iid}})$. We are now ready to state our main result of this section.






\begin{proposition}\label{prop:minimax-regret-aff}
Let $m_{n}^*$ be a SPA with reserve CDF $\Phi_{n}^*$ with the corresponding $r_{n}^*$ as given by Theorem~\ref{thm:minimax-lambda-regret-main}, and $F_{n}^*$ be the $n$ i.i.d.~isorevenue distributions with $r_{n}^*$ as given in \eqref{eq:isorev}.\obcomment{This is now in the appendix} Then for any exchangeable and affiliated distribution $\mathbf{F}$ over $n$ agents and any mechanism $m$, $R_{n}(m_{n}^*,\mathbf{F}) \leq R_{n}(m_{n}^*,\mathbf{F}_{n}^*) \leq R_{n}(m, \mathbf{F}_{n}^*)$. Therefore, $\mathcal{R}_{n}^*(\mathcal{F}_{\emph{aff}})  = \mathcal{R}_{n}^*(\mathcal{F}_{\emph{iid}})$.
\end{proposition}

Before we proceed to the formal proof, we give some high-level idea on the forces behind the result. Fixing the marginal distributions, increasing dependence should increase revenue of a SPA because the second highest valuation would tend to be closer to the highest valuation. One also needs to understand the impact of increasing dependence on the benchmark.  For affiliated distributions, we show in Lemma~\ref{lem:F-n-key} below that  the distribution of the first-order statistic is first-order dominated by the maximum of $n$ i.i.d.~copies.  Our result uses this to establish that the worst-case distribution of Nature is i.i.d.

\begin{proof}[Proof of Proposition~\ref{prop:minimax-regret-aff}]
We have already shown in the second part of the proof of Theorem~\ref{thm:minimax-lambda-regret-main} that $R_n(m_n^*,\mathbf{F}_n^*) \leq R_n(m, \mathbf{F}_n^*)$, so we only need to show the first part, which is $R_n(\Phi_{n}^*,\mathbf{F}) \leq R_{n}(\Phi_{n}^*,\mathbf{F}_n^*)$.  In the proof, we will work with order statistics of the distribution over a subset of agents, defined as follows. 
\begin{definition}
Let $\mathbf{F}$ be an exchangeable distribution and $1 \leq k \leq n$. Any restriction of $\mathbf{F}$ to $k$ valuations leads to the same distribution regardless of which $k$ specific valuations we choose. We can then define $\mathbf{F}_{k}^{(j)}$ to be the distribution of the $j$-th highest order statistics under such a restricted distribution. For example, $\mathbf{F}_{k}^{(1)}(p) = \Pr_{\mathbf{v} \sim \mathbf{F}} \left( \max(v_1,v_2,\dots,v_k) \leq p \right)$.
\end{definition}

We first prove the following lemma.
\begin{lemma}\label{lem:F-n-key}
    If $\mathbf{F}$ is exchangeable and affiliated, then $F_{n}^{(1)}(p)^{1/n} \geq F_{n-1}^{(1)}(p)^{1/(n-1)}$ for every $p$.
\end{lemma}
\begin{proof}[Proof of Lemma~\ref{lem:F-n-key}]
Theorem 23 in the Appendix of \cite{MilgromWeber82} implies that if $\mathbf{v} \sim \mathbf{F}$ such that $\mathbf{F}$ is affiliated, and $h_1,\dots,h_n$ are nondecreasing functions of $\mathbf{v}$, then $h_1(\mathbf{v}),\dots,h_n(\mathbf{v})$ are affiliated. Therefore, if we define the random variables $Y_k = \1(v_k > p)$ for each $k$, then $\mathbf{Y} = (Y_1,\dots,Y_n)$ are affiliated. Because $\mathbf{v}$ is exchangeable, so is $\mathbf{Y}$. By exchangeability, we can define for each $0 \leq k \leq n$ and $\mathbf{y} \in \{0,1\}^n$, $\Pr(\mathbf{Y} = \mathbf{y}) = u_k$ if $\mathbf{y} = (y_1,\dots,y_n)$ has $k$ 1's and $n-k$ 0's. We use the notation $1_k$ for $1$ repeated $k$ times, and likewise for $0_k$. The affiliation inequality for $\mathbf{Y}$ implies, for $0 \leq k \leq n-1$, 
\begin{align*}
    \Pr(\mathbf{Y} = (1_k,0_{n-k})) \Pr(\mathbf{Y} = (0_k,1,0_{n-k-1}) ) \leq \Pr(\mathbf{Y} = (1_{k+1},0_{n-k-1})) \Pr(\mathbf{Y} = 0_{n}),
\end{align*}
or $u_k u_1 \leq u_{k+1} u_0$. We can apply this inequality inductively to get $u_k u_0^{k-1} \geq u_1^k$ for every $1 \leq k \leq n$. From $\sum_{\mathbf{y}} \Pr(\mathbf{Y} = \mathbf{y}) = 1$ we have $u_0 + \sum_{k=1}^{n} \binom{n}{k} u_k = 1$ because there are $\binom{n}{k}$ possible $\mathbf{y}$'s that have $k$ 1's. Therefore,
\begin{align*}
    F_n^{(1)}(p)^{n-1} = u_0^{n-1} = u_0^{n-1} \left( u_0 + \sum_{k=1}^{n} \binom{n}{k} u_k \right) \geq u_0^n + \sum_{k=1}^{n} \binom{n}{k} u_0^{n-k} u_1^k = (u_0+u_1)^{n} = F_{n-1}^{(1)}(p)^{n}.\quad \qedhere
\end{align*}
\end{proof}

Back to the main proof. Using the regret expression\footnote{To prevent confusion, we remind the reader that $\mathbf{F}_n^{(1)}$ is the CDF of the first order statistics of $\mathbf{F}$ over $n$ agents, whereas $\mathbf{F}_{n}^*$ is the i.i.d.~ isorevenue distribution starting at $r_{n}^*$ and is not related to the generic $\mathbf{F}$. } from Proposition~\ref{prop:regret-expressions-general}, we have
\begin{align*}
    R(\Phi_{n}^*,\mathbf{F}) \leq r_{n}^* + \int_{p=r_{n}^*}^{p=1} \left(1-p (\Phi_{n}^*)'(p) - \Phi_{n}^*(p) \right) \left(1-\mathbf{F}_{n}^{(1)}(p) \right) dp +  \Phi_n^*(p) \left( \mathbf{F}_n^{(2)}(p) - \mathbf{F}_n^{(1)}(p) \right) dp.
\end{align*}
\cite{DavidNagaraja03-order-statistics}, equation 5.3.2, states that the distribution of the second order statistics satisfies (in our notation) 
\begin{align*}
    \mathbf{F}_n^{(2)}(p) = n \mathbf{F}_{n-1}^{(1)}(p) - (n-1) \mathbf{F}_n^{(1)}(p) .
\end{align*} The defining ODE \obcomment{For now, this is in the APPENDIX...} of $\Phi_{n}^*$ is $$(\Phi_{n}^*)'(p) = \frac{1}{p} - \frac{(n-1)r_{n}^*}{p(p-r_{n}^*)} \Phi_{n}^*(p) .$$ Substituting the expressions for $\mathbf{F}_n^{(2)}(p)$ and $(\Phi_{n}^*)'(p)$ in the inequality above gives
\begin{align*}
    R(\Phi_{n}^*,\mathbf{F}) &\leq r_{n}^* + \int_{p=r_{n}^*}^{p=1} \left( \frac{(n-1)r_{n}^*}{p-r_{n}^*} - 1 \right) \Phi_{n}^*(p) (1-\mathbf{F}_n^{(1)}(p)) + n \Phi_{n}^*(p) (\mathbf{F}_{n-1}^{(1)}(p) - \mathbf{F}_n^{(1)}(p)) \\
    &= r_{n}^* + \int_{p=r_{n}^*}^{1} \left( \frac{(n-1)r_{n}^*}{p-r_{n}^*} -1 - \frac{(n-1)p}{p-r_{n}^*} \mathbf{F}_{n}^{(1)}(p) + n \mathbf{F}_{n-1}^{(1)}(p) \right) \Phi_{n}^*(p) dp
\end{align*}

Now we apply Lemma~\ref{lem:F-n-key} which gives $\mathbf{F}_{n-1}^{(1)}(p) \leq \mathbf{F}_{n}^{(1)}(p)^{(n-1)/n}$, so
\begin{align*}
    R(\Phi_{n}^*,\mathbf{F}) &\leq r_{n}^* + \int_{p=r_{n}^*}^{1} \left( \frac{(n-1)r_{n}^*}{p-r_{n}^*} -1 - \frac{(n-1)p}{p-r_{n}^*} \mathbf{F}_{n}^{(1)}(p) + n \mathbf{F}_{n}^{(1)}(p)^{(n-1)/n} \right) \Phi_n^*(p) dp \\
    &\leq r_{n}^* + \int_{p=r_{n}^*}^{1}   \sup_{z \in [0,1]} \left\{ \frac{(n-1)r_{n}^*}{p-r_{n}^*} -1 - \frac{(n-1)p}{p-r_{n}^*} z + n z^{(n-1)/n} \right\} \Phi_n^*(p) dp.
\end{align*}
The derivative of the expression inside the supremum with respect to $z$ is $-\frac{(n-1)p}{p-r_{n}^*} + (n-1) z^{-1/n}$ which is less than or equal to zero for $z \leq \left( \frac{p-r_n^*}{p} \right)^n = F_{n}^*(p)^n$ and greater than or equal to zero for $z \geq F_{n}^*(p)^n$. Therefore, the expression is maximized when $z = F_{n}^*(p)^n$, giving us
\begin{align*}
    R_{n}(\Phi_{n}^*,\mathbf{F}) \leq r_{n}^* + \int_{p=r_{n}^*}^{1} \left( \frac{(n-1)r_{n}^*}{p-r_{n}^*} -1 - \frac{(n-1)p}{p-r_{n}^*} F_{n}^*(p)^n + n F_{n}^*(p)^{n-1} \right) \Phi_n^*(p) dp = R_{n}(\Phi_{n}^*, \mathbf{F}_{n}^*)
\end{align*}
as desired.
\end{proof}

\subsection{Mixture of I.I.D.~Distributions}

In this section, we consider another common class of distributions that captures positive dependence. We consider the class $\mathcal{F}_{\text{mix}}$ of all mixtures of i.i.d.~distributions over $n$ agents. The mixture of i.i.d.~class is very common in statistics as a way to flexibly model dependences through latent variables \citep{McLachlanLeeRathnayake19}, which is interpretable in our mechanism design context as unknown random states of the world. 

\begin{definition}
$\mathcal{F}_{\text{mix}}$ is the class of all distributions $\mathbf{F}$ such that there exists a probability measure $\mu$ on the set of distributions satisfying $\mathbf{F}(\mathbf{v}) = \int \prod_{i=1}^{n} G(v_i) d\mu(G)$ for every $\mathbf{v}$. 
\end{definition}


In the previous definition, $G$ is the i.i.d.~distribution that $n$ agent valuations are drawn i.i.d.~from, and $\mu(G)$ is the weight of $G$ in the mixture that comprises $\mathbf{F}$.  We  note that mixtures of i.i.d.~have pairwise non-negative covariance, i.e., if $\mathbf{v} \sim \mathbf{F} \in \mathcal{F}_{\text{mix}}$, then $\Cov(v_i,v_j) \geq 0$ for every $i,j$.\footnote{Proof: denote by $m(G)$ the mean for marginal $G$, then by the law of total covariance, $\Cov(v_i,v_j) = \Var_{G}(m(G)) \geq 0$.} Therefore, the class $\mathcal{F}_{\text{mix}}$ also captures positive dependence.

We note that the classes $\mathcal{F}_{\text{aff}}$ and $\mathcal{F}_{\text{mix}}$, while both capturing positive dependence,  are not ``comparable'' in the sense that one is not contained in the other, as formalized below. 
\begin{lemma}\label{lem:aff-mix}
We have that $\mathcal{F}_{\text{aff}} \subsetneq \mathcal{F}_{\text{mix}}$ and $\mathcal{F}_{\text{mix}} \subsetneq \mathcal{F}_{\text{aff}}$.
\end{lemma}
In the proof of Lemma \ref{lem:aff-mix}, 
 we construct an explicit example of a distribution that is affiliated but not a mixture of i.i.d., and another example of a  distribution that is a mixture of i.i.d.~but not affiliated.

Clearly $\mathcal{F}_{\text{iid}} \subset \mathcal{F}_{\text{mix}}$ and hence $\mathcal{R}_{n}^*(\mathcal{F}_{\text{mix}})  \ge \mathcal{R}_{n}^*(\mathcal{F}_{\text{iid}})$. The main result of this section is:

\begin{proposition}\label{prop:minimax-mix-eq-iid}
$\mathcal{R}_{n}^*(\mathcal{F}_{\emph{mix}})  = \mathcal{R}_{n}^*(\mathcal{F}_{\emph{iid}})$.
\end{proposition}

This equality holds because the class of mixtures $\mathcal{F}_{\text{mix}}$ is the convex hull of $\mathcal{F}_{\text{iid}}$. Therefore, if we interpret the minimax problem as a game between Nature and the seller, under $\mathcal{F}_{\text{mix}}$, Nature can choose a strategy randomizing over $\mathcal{F}_{\text{iid}}$ whereas under $\mathcal{F}_{\text{iid}}$ Nature must choose a pure strategy. However, because Nature goes second in the sequential game, there is an optimal pure strategy. This follows because Nature's payoff is linear in its action, i.e., the distribution. The formal proof is given in Appendix~\ref{app:sec:extend-other-F}.


\subsection{Exchangeable Distributions and All Distributions}\label{subsec:exc-all-dists} 

In this section, we consider the class $\mathcal{F}_{\text{exc}}$ of all exchangeable distributions and the class $\mathcal{F}_{\text{all}}$ of all distributions. Clearly $\mathcal{F}_{\text{exc}} \subset \mathcal{F}_{\text{all}}$ and hence $\mathcal{R}_{n}^*(\mathcal{F}_{\text{all}})  \ge \mathcal{R}_{n}^*(\mathcal{F}_{\text{exc}})$. The main result of this section is that $\mathcal{R}_{n}^*(\mathcal{F}_{\text{exc}}) = \mathcal{R}_{n}^*(\mathcal{F}_{\text{all}})= 1/e$, which is independent of $n$. This result is a slight generalization (in the case of one item) of \cite{KocyigitKR20-new}, who proved a similar result on $\mathcal{F}_{\text{all}}$. In the present case, the result emerges from the same framework as the rest of the paper. The proof is given in Appendix~\ref{app:sec:extend-other-F}.\footnote{We also omit the subscripts of $\mathbf{F}^*, m^*, \Phi^*, r^*$ in the proposition and proof to reduce conflation with the analogous quantities from the i.i.d.~case and to reduce clutter.} 


\begin{proposition}\label{prop:minimax-regret-arb}
For each $k = 1,\dots,n$, let $\mathbf{F}^{*,k}$ be a distribution on $[0,1]^n$ such that for $\mathbf{v} \sim \mathbf{F}^{*,k}$, $v_i = 0$ for any $i \neq k$, and $v_k$ follows the isorevenue distribution on $[1/e,1]$, that is,
\begin{align*}
    F^{*,k}(v_k) = \begin{cases}
    0 &\text{ if } v_k \in [0, 1/e] \\
    1 - \frac{1}{e v_k} &\text{ if } v_k \in [1/e,1) \\
    1 &\text{ if } v_k = 1.
    \end{cases}
\end{align*}
Let $\mathbf{F}^{*} = \frac{1}{n} \sum_{k=1}^{n} \mathbf{F}^{*,k}$ be the uniform mixture of all the $\mathbf{F}^{*,k}$'s.
Let $m^* = \emph{SPA}(\Phi^*)$ such that $\Phi^*$ is supported on $[1/e,1]$ and is given by
\begin{align*}
    \Phi^*(p) =  \log(v)  + 1 &\text{ for } p \in [1/e,1]
\end{align*}
Then, for any mechanism $m$ and distribution $\mathbf{F} \in \mathcal{F}_{\text{all}}$, $R_{n}(m^*,\mathbf{F}) \leq R_{n}(m^*,\mathbf{F}^{*}) \leq R_{n}(m,\mathbf{F}^{*})$. Moreover, $\mathcal{R}_{n}^*(\mathcal{F}_{\emph{exc}}) = \mathcal{R}_{n}^*(\mathcal{F}_{\emph{all}}) = R_{n}(m^*,\mathbf{F}^{*}) = 1/e$ because $\mathbf{F}^*$ is exchangeable.
\end{proposition}



This proposition shows that for any $n$, the minimax regret with respect to either $\mathcal{F}_{\text{all}}$ or $\mathcal{F}_{\text{exc}}$ is $1/e$. This is the same value as the minimax regret with respect to $\mathcal{F}_{\text{iid}}$ with $n=1$ in Corollary~\ref{cor:lambda-n-1}. Furthermore, the optimal mechanism in this case is the same as the optimal mechanism for $\mathcal{F}_{\text{iid}}$ with $n = 1$, and the worst-case distribution is of the form: select an agent uniformly at random such that this agent has the same isorevenue distribution as the $n=1$ case, and all the other agents have zero valuation. These are the formal correspondences that we previously alluded to when we say that under $\mathcal{F}_{\text{all}}$ and $\mathcal{F}_{\text{exc}}$, Nature eliminates all competition and any $n$ reduces to $n = 1$. 


\section{Minimax Optimal Mechanism and Performance: General Insights}\label{sec:char-regret-ratio}

Now that we have derived the optimal mechanism, the optimal performance, and the worst-case distribution for the minimax regret, in this section, we explore the structure of the optimal mechanism, compare the performance with ``natural'' alternative mechanisms, and evaluate the value of competition. We do so for the cases of i.i.d., exchangeable and affiliated, and mixtures of i.i.d.~distributions. 



\subsection{Structure of the Optimal mechanism}

Through our analysis, we obtain the first minimax regret results for an arbitrary number of buyers across a spectrum of distribution classes.  In  Figure \ref{fig:phi-and-F-regret-fn-n}, we illustrate the saddle point. On the left part of the figure, we depict the distribution of the reserve price $\Phi_n^*$ associated with the saddle point, and on the right part, we depict the corresponding worst-case distribution $F_n^*$.   We observe that as $n$ increases, the distribution of reserves decreases in the first-order stochastic dominance sense. This is expected because, as competition increases, the role of reserve prices diminish. 
  We also note that the seller sets reserve prices with positive probability even in the limit as competition increases. As we will discuss later, Nature can counter the seller by limiting the effective competition, even if the number of buyers is large.  In turn, the seller continues to use positive reserve prices.

\begin{figure}[h!]
	\begin{subfigure}{0.48\linewidth}
    	\centering
    	\includegraphics[height=0.75\linewidth]{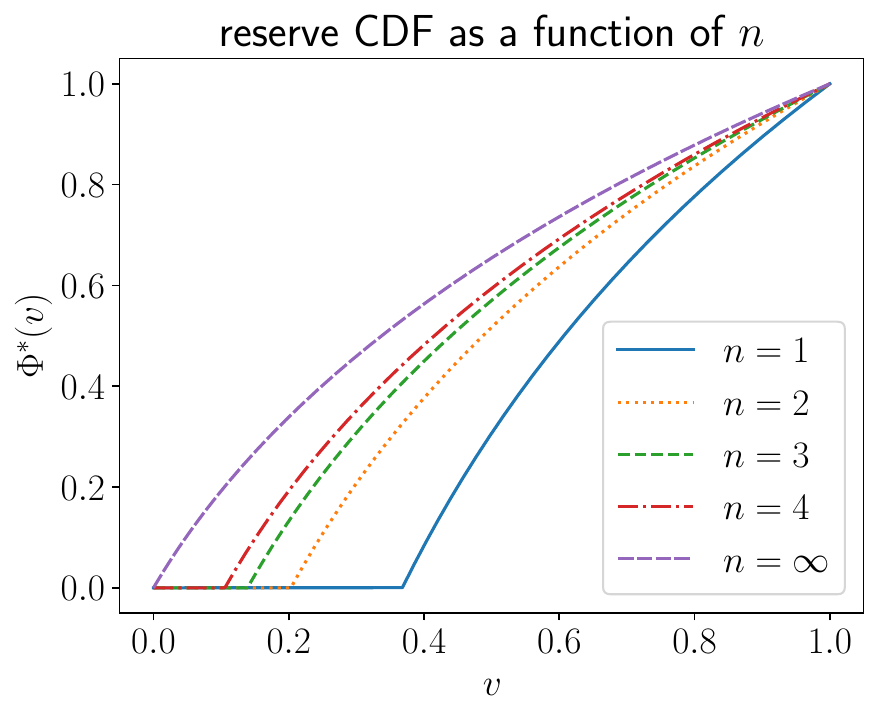}
    \end{subfigure}%
	\begin{subfigure}{0.48\linewidth}
		\centering
		\includegraphics[height=0.75\linewidth]{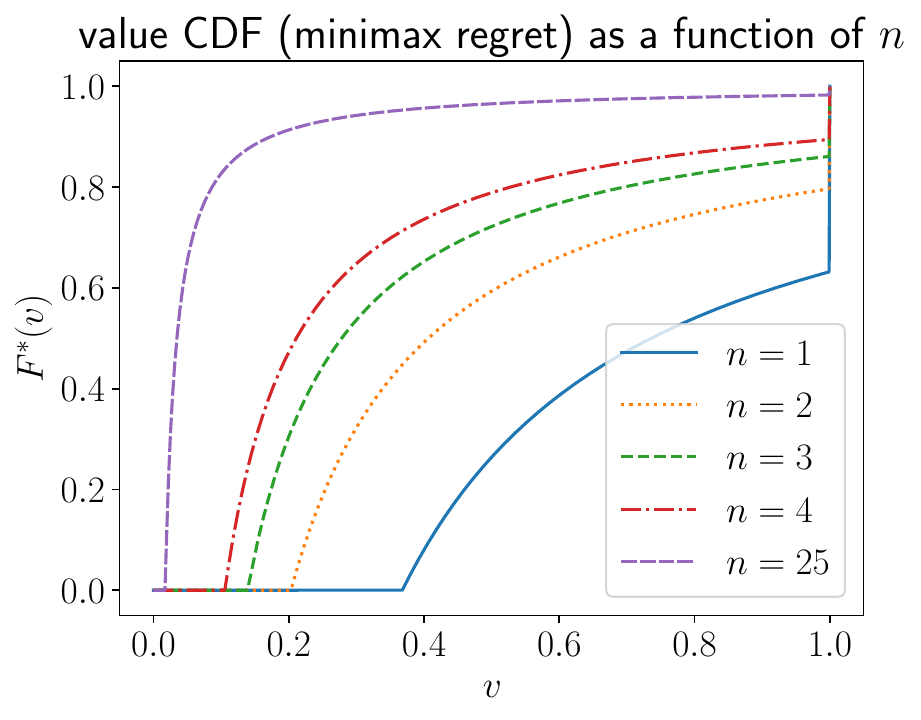}
	\end{subfigure}
	\caption{Optimal reserve CDF $\Phi_n^*$ and worst-case distribution $F_n^*$ for minimax regret as a function of $n$.}
\label{fig:phi-and-F-regret-fn-n}
\end{figure}


\subsection{Optimal Worst-Case Performance versus Alternative Mechanisms}

\subsubsection{Performance of Alternative Mechanisms}

It is also instructive to compare the minimax regret to the worst-case regret under ``natural'' mechanisms. We will consider two such mechanisms: SPA without reserve, and SPA with optimal deterministic reserve. \obedit{We derive} the worst-case regret under these mechanisms in the following proposition, whose proof is given in Appendix~\ref{app:sec:char-regret-ratio}.

\begin{proposition}\label{prop:sup-regret-spa-r}
The following results hold for all of $\mathcal{F}_{\emph{iid}}, \mathcal{F}_{\emph{aff}}, \mathcal{F}_{\emph{mix}}$. For $n \geq 2$, the worst-case regret of $\emph{SPA}(r)$, the second-price auction with fixed deterministic reserve $r$, is $\frac{(1-r)^n (n-1)^{n-1}}{((1-r)n-r)^{n-1}}$ for $r \leq \frac{1}{2}$ and $r$ for $r \geq \frac{1}{2}$. For $n = 1$, the worst-case regret is $\max(1-r, r)$. Therefore, the worst-case regret of the SPA with no reserve $\emph{SPA}(0)$ is $1$ for $n=1$ and $\left(\frac{n-1}{n}\right)^{n-1}$ for $n \geq 2$. If we can choose the reserve $r$ to minimize the worst-case regret, the optimal reserve is $r^* = \frac{1}{n+1}$ with regret $\left( \frac{n}{n+1} \right)^{n}$. 
\end{proposition}

As expected, the worst-case regret of the optimal SPA with deterministic reserve with $n$ agents is greater than the worst-case regret for SPA with no reserve with $n$ agents. Moreover, note that as $n \to \infty$, the optimal reserve price converges to zero, and the limiting worst case regret with and without reserve converge to the same value of $1/e$, which is the regret over all distributions.

In contrast to the optimal mechanism, there is no single optimal worst-case distribution for the optimal second-price auction with reserve but, instead, a sequence of worst-case distributions whose regret converges to the worst-case value. This follows because the presence of a fixed reserve price introduces a discontinuity in the revenue of a second-price auction as a function of the distribution. To wit, in the $\mathcal{F}_{\text{iid}}$ case, for the optimal reserve price $r^* = 1/(n+1)$ and $\epsilon>0$ small enough, consider the two-point distribution that puts mass on 1 with probability $1/(n+1)$ and on $r^*-\epsilon$ with probability $n/(n+1)$. The regret of this sequence of distributions converges to the optimal regret as $\epsilon \downarrow 0$. Intuitively, Nature counters the second-price auction by placing an atom immediately below the reserve price to maximize the revenue loss when the reserve price is not cleared and another at 1 to maximize the difference between the first- and second-highest valuations.

Interestingly, the worst-case regret of a SPA with optimal deterministic reserve with $n$ agents is equal to the worst-case regret for SPA with no reserve with $n+1$ agents. This result has a similar flavor with the classic result on auctions versus negotiations of \cite{BulowKlemperer96}.\footnote{An important difference is that in our setting, the worst-case distribution changes (Nature is allowed to react to the mechanism) whereas in \cite{BulowKlemperer96} the distribution is fixed.} 

\subsubsection{Comparison} \label{sec:comparison}


Next, we compare, for each value of $n$,  the optimal worst-case regret to that of $\text{SPA}(0)$ and a SPA with optimal deterministic reserve price. 
The results are presented in  Table~\ref{table:minimax-regret-with-spa}.
\begin{table}[h!]
\centering
\begin{tabular}{ c | c | c | c }
 $n$ & OPT & $\text{SPA}(0)$ & $\text{SPA}(r^*)$ \\ \hline
 1 & 0.3679 & 1.0000 & 0.5000 \\
 2 & 0.3238 & 0.5000 & 0.4444 \\
 3 & 0.3093 & 0.4444 & 0.4219 \\
 4 & 0.3021 & 0.4219 & 0.4096 \\
 5 & 0.2979 & 0.4096 & 0.4019 \\
10 & 0.2896 & 0.3874 & 0.3855 \\
25 & 0.2847 & 0.3754 & 0.3751 \\
$\infty$ & 0.2815 & 0.3679  & 0.3679
\end{tabular}
\caption{Worst-case regret under the optimal mechanism in comparison with SPA for each $n$. The columns are: OPT is the optimal mechanism, $\text{SPA}(0)$ is the SPA with no reserve, and $\text{SPA}(r^*)$ is the SPA with optimal deterministic reserve.}
\label{table:minimax-regret-with-spa}
\end{table}


We observe that despite the fact that Nature counters competition, as discussed earlier, there is significant scope for optimization  and the worst-case performance  can be significantly improved when moving from $\text{SPA}(0)$ or  $\text{SPA}(r^*)$ to the optimal minimax mechanism. For example, with three buyers, the optimal mechanism leads to  an improvement of more than 25\% compared to the alternative mechanisms. 

\subsection{On the Value of Competition in a Minimax Regret Framework}\label{subsec:value-of-competition-minimax-regret}


The previous section shows the minimax regret for different values of $n$. As expected, the regret decreases as $n$ increases, which highlights the  the value of competition for the seller. However, we observe that the value of competition has diminishing returns. For example, in Table~\ref{table:minimax-regret-with-spa} we can see that the minimax regret from $n = 1$ to $n= 2$ decreases by 12\%; from $n = 2$ to $n = 3$, 4.5\%; from $n = 3$ to $n  = 4$, 2.3\%; from $n = 4$ to $n = 5$, 1.4\%. 
We also note that the limiting minimax regret as $n \to \infty$ is strictly positive. This is formally captured in the following corollary that characterizes minimax regret performance  as $n \to \infty$.

\begin{corollary}\label{cor:lambda-n-infty}
We have $\lim_{n \to \infty} n r_{n}^* = c$ where $c \approx 0.434818$ is a solution to $e^{-c} -  \int_{1}^{\infty} \frac{e^{-cx}}{x} dx = 0$. The regret   $R_{n}(\Phi_{n}^*,F_{n}^*)$ converges to a constant $e^{-c} -  \int_{0}^{1} e^{-c/v} dv \approx 0.281494$. Furthermore, the optimal mechanism
$\Phi_{n}^*$ converges in distribution to $\Phi_{\infty}^*$ given by 
\begin{align*}
    \Phi_{\infty}^*(v) = \int_{1}^{\infty} \frac{1}{x} \exp \left( - \frac{c}{v}(x-1) \right) dx \text{ for } v \in [0,1].
\end{align*}
\end{corollary}


It might seem surprising at first that increasing competition in a setting with bounded i.i.d.~valuations does not lead to a vanishing regret as competition increases. However, it is important to note the following: $i.)$ Nature can select a distribution that indirectly limits the ``effective" competition when $n$ is large and $ii.)$ the benchmark in the regret is the best revenues one could obtain without information asymmetry. Indeed, in the present framework, Nature selects a distribution that depends on $n$.  It is possible to show that asymptotically the distribution associated with the saddle point puts mass of order $\Theta(1/n)$ at 1, and significant mass around zero. For large $n$, under  the optimal mechanism, the fundamental contribution to regret under the worst-case distribution is as follows: with high probability, there is at most one agent that has a valuation significantly different from zero, so the mechanism cannot use other agents as a competitive benchmark. In other words, the choice of the worst-case distribution negates a large part of the effect of competition. A related phenomenon explains the  behavior of the worst-case performance under the alternative mechanisms. Finally, while the choice of the benchmark impacts the limiting value of regret, the observation that regret does not vanish as the number of buyers increases is robust to the choice of the benchmark. For example, in Proposition~\ref{prop:minimax-regret-opt-limit} we show the minimax regret against the second-best benchmark (optimal revenues with knowledge of the distribution of values) is at least $e^{-2} \approx 0.1353$, regardless of the number of buyers.













\section{Conclusion}\label{sec:conclude}

In this paper, we analyze the fundamental problem of optimally selling an item to an arbitrary number $n$ of buyers without the common prior assumption. Knowing only \obedit{an upper bound on} \obdelete{the support of} the buyers' valuations, the seller selects a DSIC mechanism that minimizes the worst-case regret between the mechanism revenue and the benchmark revenue when the buyers' valuations are known. We consider a broad range of classes of admissible distributions: i.i.d, mixtures of i.i.d., exchangeable and affiliated, exchangeable, and all distributions. We show that the first three have minimax regret values that depend on $n$ and exactly characterize the optimal mechanism and performance for each $n$, while the last two reduce to the case $n = 1$, i.e., the no-competition case. Taken together, these results show the strength (or lack thereof) of different distributional class assumptions and quantify the value of competition in the minimax framework. 

There are many natural directions for future work. One potential direction is to consider other benchmarks, especially the second-best benchmark defined by the maximum revenue achievable by the optimal mechanism that knows the distribution $\mathbf{F}$ but not the buyers' valuations. This benchmark captures a different layer of information asymmetry: it captures the value of knowing the buyers' valuation distribution, whereas this paper's benchmark captures the value of knowing the buyers' valuations. We have discussed the first-best versus second-best benchmarks in Section~\ref{subsec:intro-model-discuss} and Appendix~\ref{app:subsec:model-discuss-benchmark}, and  we expect that novel methodological approaches are required to deal with the second-best benchmark.

Another direction is to consider the impact of side information. In this paper, we consider the setting with minimal knowledge (only the upper bound on the valuations is known). In practice, while we do not have complete knowledge about the environment, we often do have \textit{partial knowledge}. This partial knowledge about the distribution might be the mean or other moments (as in \citep{BachrachTalgamCohen22,Che22,shixin-ratio,kleer-optimal-stopping-robust}), the median or other quantiles (as in \citep{AllouahBahamouBesbes-pricing-single-point}), or samples from the i.i.d.~distribution (as in \citep{making-the-most-of-your-samples,AllouahBahamouBesbes22-pricing-with-samples}), or the shape of the distribution (such as regularity or monotone-hazard-rate). 
How to incorporate additional side information into the robust framework is mostly an open question, but it can be a bridge to bring the robust theory closer to practice. \sbdelete{It can justify why certain mechanism formats might be prevalent in practice (as in this work, showing that second-price auctions are robustly optimal across environments, or robustness of linear contracts in \cite{Carroll15-robust-linear-contracts}), suggest new forms of mechanisms (such as the inflated second price auctions of \cite{fu2015randomization,AllouahBesbes20, HartlineJohnsenLi20-benchmark-design}), tell us how to optimally use different forms of information and quantify its value. The partial information paradigm can also interact with the choice of objectives in interesting ways.} 

 While in the setting studied in this work, only minimax regret is nontrivial, with other forms of partial information mentioned earlier all of minimax regret, maximin ratio, and maximin revenue are nontrivial. Developing new methodologies to unify these objectives and better understand the tradeoffs between them should be a fruitful agenda. 
 
  Another related avenue is to develop mechanisms that optimally tradeoff robust guarantees when the additional information is correct and when it is not, as in the emerging area of algorithms with predictions/advice \citep{algorithms-with-predictions-survey}.

The other natural direction is to consider different notions of robustness and equilibrium concepts for buyers. While the class of DSIC mechanisms has its advantages in robustness and simplicity (as discussed in Section~\ref{subsec:intro-model-discuss}), there are non-truthful mechanisms, such as first-price auctions, that are used in practice, and it will be useful to have a way to robustly design them as well. 



 {
 \small
 \setstretch{1.20}
 \bibliographystyle{plainnat}
 \bibliography{references}


\newpage
\setstretch{1.35}
\appendix

\pagenumbering{arabic}
\renewcommand{\thepage}{App-\arabic{page}}
\renewcommand{\theequation}{\thesection-\arabic{equation}}
\renewcommand{\thelemma}{\thesection-\arabic{lemma}}
\renewcommand{\theproposition}{\thesection-\arabic{proposition}}
\setcounter{page}{1}
\setcounter{section}{0}
\setcounter{proposition}{0}
\setcounter{lemma}{0}
\setcounter{equation}{0}

\setcounter{footnote}{0}

%

\begin{center}
 {\Large \textbf{Electronic Companion: 
\\ On the Robustness of Second-Price Auctions\\ in Prior-Independent Mechanism Design\\}
\medskip
Jerry Anunrojwong\footnote{Columbia University, Graduate School of Business. Email: {\tt janunrojwong25@gsb.columbia.edu}}, 
Santiago R. Balseiro\footnote{Columbia University, Graduate School of Business. Email: {\tt srb2155@columbia.edu}.}, ~ and Omar Besbes\footnote{Columbia University, Graduate School of Business. Email: {\tt ob2105@columbia.edu}.}.}
\end{center}

\section{Deriving Candidate Mechanisms and Distributions}\label{app:sec:derive-candidates}

In this section, we derive candidate mechanisms and distributions. To do so, we will make additional assumptions. We highlight here that this section does not directly imply that we have a saddle point, but provides a ``guided search" for a saddle point. In Section \ref{subsec:proof-main-theorem}, we provide a formal verification  that the candidates obtained here are indeed a saddle point for the original problem. 

\subsection{Deriving Candidate Saddle Points}\label{sec:deriving-candidate}

In the next result, we establish that under some conditions, the worst-case distribution in a saddle point has to satisfy a particular ordinary differential equation (ODE).

\begin{proposition}\label{prop:opt-F-ode}
Suppose that the problem $\inf_{\Phi} \left[ \sup_{\mathbf{F}} R_{n}(\Phi,\mathbf{F}) \right]$ has an interior solution $\Phi^*$ which has a density and is supported on $[r^*,1]$, and is such that the problem $\sup_{\mathbf{F}} R_{n}(\Phi^*,\mathbf{F})$ over i.i.d.~$\mathbf{F}$ has a unique solution $\mathbf{F}^*$ with marginal $F^*$. Then $F^*$ satisfies the following ODE for $v \in (r^*,1)$:
$$1-F^*(v)-v(F^*)'(v) = 0.$$
\end{proposition}

Proposition~\ref{prop:opt-F-ode} suggests that if a SPA is minimax optimal, then the worst-case $\mathbf{F}^*$ in the saddle should satisfy $1-F^*(v)-v(F^*)'(v) = 0$ for $v \in (r^*,1)$. Together with the initial condition $F^*(r^*)$ = 0, the solution to this ODE is 
\begin{align*}\label{eq:isorev}
    F^*(v) = \begin{cases}
    0 &\text{ if } v \in [0,r^*] \\
    1-r^*/v &\text{ if } v \in [r^*,1) \tag{$F$-IsoRevenue} \\
    1 &\text{ if } v = 1
    \end{cases}
\end{align*}

Note that $F^*(1^-) = \lim_{v \uparrow 1} F^*(v) = 1-r^*$, so there is a point mass of size $r^*$ at $v = 1$. This distribution corresponds to an  \textit{isorevenue} distribution because if the willingness to pay distribution follows this distribution $F^*$, then under any posted price $p \in [r^*,1)$, the revenue is the same and is given by $p(1-F^*(p)) = r^*$.

Next, we fix the distribution to be of the form above, and ask what candidate mechanisms could lead to a saddle point in the class of SPA with random reserve. In other words, is there a SPA with random reserve such that $\mathbf{F}^*$ maximizes its regret? We establish that if there is such a reserve price distribution, then it must satisfy   a particular ODE.

\begin{proposition}\label{prop:opt-phi-ode}
Fix $r^* \in [0,1]$. Let $F^*$ be given by \emph{($F$-IsoRevenue)} and $\Phi^*$ be a CDF of a distribution that has a density and is supported on $[r^*,1]$. Suppose that  $\mathbf{F}^*$ is the unique solution to $\sup_{\mathbf{F}} R_{n}(\Phi^*,\mathbf{F})$. Then $\Phi^*$ satisfies the following ODE for $v \in (r^*,1)$:
\begin{align}\label{eq:ODE-Phi}
    (\Phi^*)'(v) + \frac{(n-1)r^*}{v(v-r^*)} \Phi^*(v) = \frac{1}{v}. \tag{ODE-$\Phi$}
\end{align}
\end{proposition}

This ODE uniquely describes $\Phi^*$ as a function of $r^*$.  The final question left is the candidate value of $r^*$. The initial condition $\Phi^*(r^*) = 0$ uniquely determines the only possible value of $r^*$. 


The next result characterizes, in quasi-closed form, a solution to the  ODE to determine the functional form of $\Phi^*$, the value of $r^*$, as well as the value of the regret under the candidate saddle point. 

\begin{proposition}\label{prop:ode-phi-solution}Let $r_{n}^* \in (0,1/n)$ be the unique solution to 
\begin{align*}
    (1-r^*)^{n-1} + \log(r^*) + \sum_{k=1}^{n-1} \frac{(1-r^*)^k}{k} = 0.
\end{align*}
Then the function $\Phi_{n}^*(v)$ given by 
\begin{align*}
    \Phi_{n}^*(v) = \left( \frac{v}{v-r_{n}^*} \right)^{n-1} \log \left( \frac{v}{r_{n}^*} \right) - \sum_{k=1}^{n-1} \frac{1}{k} \left( \frac{v}{v-r_{n}^*} \right)^{n-1-k} = \sum_{k=n}^{\infty} \frac{1}{k} \left( \frac{v-r_{n}^*}{v} \right)^{k-(n-1)}
\end{align*}
satisfies  the ODE \eqref{eq:ODE-Phi} as well as  $\Phi_n^*(r_n^*) = 0$ and $\Phi_n^*(1) = 1$.

Furthermore, under the isorevenue $\mathbf{F}_{n}^*$ given by $F_{n}^*(v) = 1 - r_{n}^*/v$ for $v \in [r_{n}^*,1)$ and $F_{n}^*(1) = 1$, the regret is given by
\begin{align*}
    R_{n}(\Phi^*_{n}, \mathbf{F}_{n}^*) = (1-r_{n}^*)^{n-1} - \int_{v=r_{n}^*}^{v=1} \left( 1 - \frac{r_{n}^*}{v} \right)^{n-1} dv.
\end{align*}
\end{proposition}



The above yields a candidate saddle point for the original problem. We next provide a formal verification that the mechanism obtained minimizes the worst case regret among all mechanisms (and not only SPA mechanisms), and that the candidate distribution is indeed the worst-case distribution against this candidate mechanism.





\section{Proofs and Discussions from Section~\ref{sec:minimax-regret-main} and Appendix~\ref{app:sec:derive-candidates}}\label{app:sec:minimax-regret-main} 

\begin{proof}[Proof of Step 2 (Seller's Optimality) from Section~\ref{subsec:proof-main-theorem}]
    Because $(\mathbf{x},\mathbf{p})$ is dominant strategy incentive compatible, \cite{Myerson81} gives the following payment characterization
\begin{align*}
    p_i(v_i,\mathbf{v}_{-i}) = p_i(0,\mathbf{v}_{-i}) + v_i x_i(v_i,\mathbf{v}_{-i}) - \int_{\tilde{v}_i=0}^{\tilde{v}_i=v_i} x_i(\tilde{v}_i,\mathbf{v}_{-i}) d\tilde{v}_i
\end{align*}

Individual rationality requires that $p_i(0,\mathbf{v}_{-i}) \leq 0$. Because the seller aims to maximize revenue, it is optimal to set $p_i(0,\mathbf{v}_{-i}) = 0$. We can then write, for fixed $\mathbf{v}_{-i}$,
\begin{align*}
    \bE_{v_i \sim F_{n}^*}[ p_i(\mathbf{v}) ] &= \int_{v_i=r_{n}^*}^{v_i=1} \frac{r^*}{v_i^2} p_i(v_i,\mathbf{v}_{-i}) dv_i + r_{n}^* p_i(v_i=1,\mathbf{v}_{-i}) \\
    &= \int_{v_i=r_{n}^*}^{v_i=1} \frac{r_{n}^*}{v_i^2} \left( v_i x_i(v_i,\mathbf{v}_{-i}) - \int_{\tilde{v}_i=0}^{\tilde{v}_i=v_i} x_i(\tilde{v}_i,\mathbf{v}_{-i}) d \tilde{v}_i \right) dv_i \\
    &+ r_{n}^* \left(  x_i(v_i=1,\mathbf{v}_{-i}) - \int_{v_i=0}^{v_i=r_n^*} x_i(v_i,\mathbf{v}_{-i}) dv_i - \int_{v_i=r_{n}^*}^{v_i=1} x_i(v_i,\mathbf{v}_{-i}) dv_i \right) 
\end{align*}
We can compute the double integral as follows
\begin{align*}
    &\int_{v_i=r_{n}^*}^{v_i=1} \int_{\tilde{v}_i=0}^{\tilde{v}_i=v_i} \frac{r_{n}^*}{v_i^2} x_i(\tilde{v}_i,\mathbf{v}_{-i}) d \tilde{v}_i dv_i \\
    &= \int_{\tilde{v}_i=0}^{\tilde{v}_i=r_{n}^*} \int_{v_i=r_{n}^*}^{v_i=1} \frac{r_{n}^*}{v_i^2} x_i(\tilde{v}_i,\mathbf{v}_{-i}) dv_i d \tilde{v}_i + \int_{\tilde{v}_i=r_{n}^*}^{\tilde{v}_i=1} \int_{v_i=\tilde{v}_i}^{\tilde{v}_i=1} \frac{r_{n}^*}{v_i^2} x_i(\tilde{v}_i,\mathbf{v}_{-i}) dv_i d \tilde{v}_i \\
    &= \int_{v_i=0}^{v_i=r_{n}^*} \left( 1 - r_{n}^* \right) x_i(v_i,\mathbf{v}_{-i}) dv_i + \int_{v_i=r_{n}^*}^{v_i=1} \left( \frac{r_{n}^*}{v_i} - r_{n}^* \right) x_i(v_i,\mathbf{v}_{-i}) dv_i 
\end{align*}
Substituting this expression gives 
\begin{align*}
    \bE_{v_i \sim F_{n}^*}[ p_i(\mathbf{v})] 
    &= r_{n}^*   x_i(v_i=1,v_{-i}) - \int_{v_i=0}^{v_i=r_{n}^*} x_i(v_i,v_{-i}) dv_i
\end{align*}
Taking the expectation of the above over $\mathbf{v}_{-i} \sim (F_{n}^*)^{n-1}$ gives agent $i$'s revenue $\bE_{\mathbf{v} \sim (F_{n}^*)^{n}} [p_i(\mathbf{v})]$. Therefore, the objective is
\begin{align*}
    \sum_{i=1}^{n} \bE_{\mathbf{v} \sim (F_{n}^*)^{n}} [p_i(\mathbf{v})] = r^* \sum_{i=1}^{n} \bE_{\mathbf{v}_{-i} \sim (F_{n}^*)^{n-1}} [x_i(v_i=1,\mathbf{v}_{-i})] - \sum_{i=1}^{n} \int_{v_i=0}^{v_i=r_{n}^*} \bE_{\mathbf{v}_{-i} \sim (F_{n}^*)^{n-1}} [ x_i(v_i,\mathbf{v}_{-i}) ] dv_i .
\end{align*}

Given that $x_i$ has to satisfy the constraints $x_i(\mathbf{v}) \geq 0$, $ \forall i, \forall \mathbf{v} \in [0,1]^n$ and $\sum_{i=1}^{n} x_i(\mathbf{v}) \leq 1 \forall \mathbf{v} \in [0,1]^n$, the objective is maximized if and only if for any nonempty set $S \subseteq \{1,2,\dots,n\}$ the following holds:
\begin{enumerate}
    \item[(i)] $\text{ if } v_i = 1 \text{ for } i \in S \text{ and } v_i \in (r_{n}^*,1) \text{ for } i \not\in S \text{ then } \sum_{i \in S} x_i(\mathbf{v}) = 1 \text{ and } x_i(\mathbf{v}) = 0 \text{ for } i \not\in S$, and
    \item[(ii)] $\text{ for every } i \text{ if } v_i < r_{n}^* \text{ and } v_j > r_{n}^* \text{ for all } j \neq i \text{ then } x_i(\mathbf{v}) = 0$.
\end{enumerate}
In words, (i) if there is at least one buyer with valuation $1$ and every buyer's valuation is above $r_{n}^*$, then the mechanism allocates with probability one and can only allocate to agent(s) with valuation $1$; (ii) if a buyer's valuation is below $r_{n}^*$ and all other buyers' valuations are above $r_{n}^*$, then that buyer is never allocated. It is clear that a second-price auction with random reserve drawn from a CDF $\Phi_{n}^*$ with a density and supported on $[r_{n}^*,1]$ corresponds to an allocation rule that satisfies the above condition, so the mechanism is optimal, and the saddle inequality is proved.
\end{proof}

\begin{proof}[Proof of Proposition~\ref{prop:regret-expressions-general}]

Because $\Phi$ has a density $\phi$, the regret of $\text{SPA}(\Phi)$ at the valuation vector $\mathbf{v}$ is given by
\begin{align*}
    &  v^{(1)} -  \bE_{p \sim \Phi} \left[ \max(v^{(2)},p) \1(v^{(1)} \geq p) \right] \\
    &=  v^{(1)} -  \left(   \int_{p=0}^{p=v^{(2)}} v^{(2)} \phi(p) dp + \int_{p=v^{(2)}}^{p=v^{(1)}} p \phi(p) dp  \right) \\
    &= v^{(1)} - \left[ v^{(2)}  \Phi(v^{(2)})  + \left[ p \Phi(p) \right]_{p=v^{(2)}}^{p=v^{(1)}} - \int_{p=v^{(2)}}^{p=v^{(1)}} \Phi(p) dp \right] \\
    &= v^{(1)} - v^{(1)} \Phi(v^{(1)}) + \int_{p=v^{(2)}}^{p=v^{(1)}} \Phi(p) dp
\end{align*}

Therefore,
\begin{align*}
    R_{n}(\Phi,\mathbf{F}) = \bE_{\mathbf{v} \sim \mathbf{F}} \left[ v^{(1)} - v^{(1)} \Phi(v^{(1)}) \right] + \bE_{ \mathbf{v} \sim \mathbf{F} } \left[ \int_{p=v^{(2)}}^{p=v^{(1)}} \Phi(p) dp \right] 
\end{align*}

We will compute the two terms separately. We capture the calculation of the first term in the following lemma.
\begin{lemma}\label{lem:first-term-calc-gen}
    Let $h(v)$ be a continuous function that is differentiable everywhere except perhaps at $v=r^*$, and $G$ be a distribution on $[0,1]$, then we have
    \begin{align*}
        \int_{v \in [0,1]} h(v) dG(v) &= h(0) G(r^*) + \int_{v=0}^{v=r^*} h'(v) (G(r^*)-G(v)) dv 
        \\
        &+ h(r^*) (G(1) - G(r^*)) + \int_{v=r^*}^{v=1} h'(v) (G(1)-G(v)) dv 
    \end{align*}
\end{lemma}


\begin{proof}[Proof of Lemma~\ref{lem:first-term-calc-gen}]
    We have
    \begin{align*}
        \int_{v \in [0,r^*]} h(v) dG(v) &= \int_{v \in [0,r^*]} \left( h(0) + \int_{\tilde{v}=0}^{\tilde{v}=v} h'(\tilde{v}) d\tilde{v} \right) d G(v) 
        \\
        &= h(0) G([0,r^*]) + \int_{\tilde{v}=0}^{\tilde{v}=r^*} \int_{v=\tilde{v}}^{v=r^*} h'(\tilde{v}) dG(v) d\tilde{v} 
        \\
        &= h(0) G(r^*) + \int_{\tilde{v}=0}^{\tilde{v}=r^*} h'(\tilde{v}) (G(r^*)-G(\tilde{v})) d\tilde{v} 
    \end{align*}
    where the first line holds by the fundamental theorem of calculus, and the second line holds by Fubini. Therefore, $\int_{v \in [0,r^*]} h(v) dG(v)$ is equal to the expression in the first line of the lemma. Similarly, $\int_{v \in (r^*,1]} h(v) dG(v)$ is equal to the expression in the second line of the lemma. Finally, $\int_{v \in [0,1]} h(v) dG(v)$ is the sum of those two terms.
\end{proof}

We let $h(v) = v - v \Phi(v)$ and $G(v) = \mathbf{F}_n^{(1)}(v)$ in Lemma~\ref{lem:first-term-calc-gen} to get that the first term of the regret is
\begin{align*}
    \bE_{\mathbf{v} \sim \mathbf{F}} \left[ v^{(1)} - v^{(1)} \Phi(v^{(1)}) \right] = r^* -  \int_{v=0}^{v=r^*} \mathbf{F}_n^{(1)}(v) dv + \int_{v=r^*}^{v=1} (1 -\Phi(v) - v\Phi'(v)) (1 - \mathbf{F}_n^{(1)}(v)) dp
\end{align*}
Now we compute the second term.
\begin{align*}
    \bE_{\mathbf{v} \sim \mathbf{F}} \left[  \int_{p=v^{(2)}}^{p=v^{(1)}} \Phi(p) dp \right] 
    = \int_{\mathbf{v} \in [0,1]^n} \int_{p \in [0,1]} \Phi(p) \1(v^{(2)} \leq p < v^{(1)}) dp d\mathbf{F}(\mathbf{v}) \\
    = \int_{p \in [0,1]} \int_{\mathbf{v} \in [0,1]^n} \Phi(p) \1(v^{(2)} \leq p < v^{(1)}) d\mathbf{F}(\mathbf{v}) dp = \int_{v=r^*}^{v=1} \Phi(v) (\mathbf{F}_n^{(2)}(v) - \mathbf{F}_n^{(1)}(v)) dv
\end{align*}

where the last equality holds because $\Phi$ has support on $[r^*,1]$. Adding up the two terms, we get the (Regret-$\mathbf{F}$) expression.

Now we assume that $\mathbf{F}$ is i.i.d., then $\mathbf{F}_n^{(1)}(v) = F(v)^n$ and $\Pr_{\mathbf{F}}(v^{(2)} \leq v < v^{(1)}) = n F(v)^{n-1}(1-F(v))$, so we get the (Regret-$F$) expression.

To derive the (Regret-$\Phi$) expression that is linear in $\Phi$, we can use integration by part to transform the term involving $\Phi'(v)$ into an expression linear in $\Phi$ as follows.
\begin{align*}
    & \int_{v=r^*}^{v=1} v(1-F(v)^n) \Phi'(v) dv \\
    &= \left[ v(1-F(v)^n) \Phi(v) \right]_{v=0}^{v=1} - \int_{v=r^*}^{v=1} \frac{d}{dv} (v (1-F(v)^n) ) \Phi(v) dv  \\
    &= -\int_{v=r^*}^{v=1} (1-F(v)^n - nvF(v)^{n-1} F'(v)) \Phi(v) dv
\end{align*}
Plugging this into the (Regret-$F$) expression, we get the (Regret-$\Phi$) expression as desired.

\end{proof}

\begin{proof}[Proof of Proposition~\ref{prop:opt-F-ode}]
Define $h(\Phi) = \sup_{\mathbf{F}} R_{n}(\Phi,\mathbf{F})$ to be the worst-case regret for a distribution of reserves $\Phi$ and $Z_0(\Phi) = \{\tilde{\mathbf{F}}: R_{n}(\Phi,\tilde{\mathbf{F}}) = \sup_{\mathbf{F}} R_{n}(\Phi,\mathbf{F}) \}$ the set of optimal worst-case distributions for $\Phi$. From the (Regret-$\Phi$) expression given in Proposition~\ref{prop:regret-expressions-general}, $R_{n}(\Phi,\mathbf{F})$ is linear. Because the pointwise supremum of linear functions is convex, we have that $h(\Phi)$ is convex in $\Phi$.  Since we assumed that $Z_0(\Phi^*) = \{\mathbf{F}^*\}$, Danskin's theorem then implies that $h$ is differentiable at $\Phi^*$, and the derivative of $h$ at $\Phi^*$ is given by $\frac{\partial h}{\partial \Phi}(\Phi^*) = \frac{\partial R_{n}}{\partial \Phi} (\Phi^*,\mathbf{F}^*).$ Therefore, we must have $\frac{\partial h}{\partial \Phi}(\Phi^*) = 0$ since $h$ is convex and has a minimum at $\Phi^*$, which implies that $\frac{\partial R_{n}}{\partial \Phi} (\Phi^*,\mathbf{F}^*)= 0$. Because $R_{n}(\Phi,F^*)$ is linear in $\Phi$, this means the coefficient of $\Phi(v)$ in $R_{n}(\Phi,\mathbf{F}^*)$ should be zero for $v \in (r^*,1)$, which implies $1-F^*(v)-v(F^*)'(v) = 0$ as desired.
\end{proof}


\begin{proof}[Proof of Proposition~\ref{prop:opt-phi-ode}]
We note that the (Regret-$F$) expression for $R_{n}(\Phi^*,\mathbf{F})$
\begin{align*}
    r^*     - \int_{v=0}^{v=r^*} F(v)^n dv \nonumber + \int_{v=r^*}^{v=1} (\lambda -v (\Phi^*)'(v) - \Phi(v) ) (1-F(v)^n ) +  \Phi^*(v) n F(v)^{n-1}(1-F(v)) dv
\end{align*}
involves $F(v)$ only. To maximize this expression we can set $F(v) = 0$ for $v \in [0,r^*]$. We can further maximize the integrand for $v \in [r^*,1]$ pointwise. We have the constraint that $F(v)$ must be nondecreasing, but because the unique optimal solution is $F^*$ which is strictly increasing, this monotonicity constraint must not bind, and the pointwise optimization procedure must give $F^*$ as a solution.  
Taking the derivative of the expression under the integral with respect to $F(v)$ gives
\begin{align*}
    (1-\Phi^*(v) - v (\Phi^*)'(v)) (-nF(v)^{n-1}) + n \Phi^*(v) ( (n-1) F(v)^{n-2} - n F(v)^{n-1} ) = 0  \\
    (1+(n-1) \Phi^*(v) - v (\Phi^*)'(v))(1-F(v)) -1 + v(\Phi^*)'(v) = 0
\end{align*}

From the previous subsection, we want $\Phi^*$ to be such that the optimal $F$ from this is $F^*(v) = 1 - \frac{r^*}{v}$ for $v \geq r^*$. Substituting this gives
\begin{align}
    (\Phi^*)'(v) + \frac{(n-1)r^*}{v(v-r^*)} \Phi^*(v) = \frac{\lambda}{v} \tag{ODE-$\Phi$}
\end{align}
as desired.\end{proof}

\begin{proof}[Proof of Proposition~\ref{prop:ode-phi-solution}]

We will sometimes omit the star and the subscripts for simplicity. We solve the ODE with the integrating factor method. With the ODE of the form $\Phi' + g \Phi = h$ where $g $ and $h$ do not depend on $\Phi$, we multiply both sides by the integrating factor $I := \exp \left( \int g dv \right)$ which has the property $I' = gI$, so $(I\Phi)' = I \Phi' + gI \Phi = I h$. In our case, the integrating factor is 
\begin{align*}
    \exp \left( \int \frac{(n-1)r^*}{v(v-r^*)} dv \right) = \exp \left( \int \frac{(n-1)}{v-r^*} - \frac{(n-1)}{v} dv \right) = \left( \frac{v-r^*}{v} \right)^{n-1}
\end{align*}
If we multiply by $\left( \frac{v-r^*}{v} \right)^{n-1}$ on both sides we get
\begin{align*}
    \frac{d}{dv} \left[ \left( \frac{v-r^*}{v} \right)^{n-1} \Phi(v) \right] = \frac{1}{v} \left( \frac{v-r^*}{v} \right)^{n-1}
\end{align*}

Now we integrate on both sides. Using the change of variable $w = (v-r^*)/v$ we can integrate the right hand side as follows
\begin{align*}
    \int \frac{1}{v} \left( \frac{v-r^*}{v} \right)^{n-1} dv = \int \frac{w^{n-1}}{1-w} dw = \text{const} -\log(1-w) - \sum_{k=1}^{n-1} \frac{w^k}{k}
\end{align*}


Therefore, there is a constant $c$ such that
\begin{align*}
    \left( \frac{v-r^*}{v} \right)^{n-1} \Phi(v) =  c +  \log \left( \frac{v}{r^*} \right) - \sum_{k=1}^{n-1} \frac{1}{k} \left( \frac{v-r^*}{v} \right)^{k}
\end{align*}

We will now determine the constant $c$. Using $\Phi(r^*) = 0$, letting $v \downarrow r^*$ we get that the left hand side goes to 0 and the right hand side goes to $c$, so $c = 0$, and we have
\begin{align*}
    \Phi(v) = \left( \frac{v}{v-r^*} \right)^{n-1} \log \left( \frac{v}{r^*} \right) - \sum_{k=1}^{n-1} \frac{1}{k} \left( \frac{v}{v-r^*} \right)^{n-1-k}
\end{align*}


We use the Taylor series for $\log(1+x)$ with $x = (v-r^*)/r^*$ to get the other expression for $\Phi(v)$:
\begin{align*}
    \Phi(v) =  \sum_{k=n}^{\infty} \frac{1}{k} \left( \frac{v-r^*}{v} \right)^{k-(n-1)}
\end{align*}

We pin down the value of $r^*$ by requiring that $\Phi(1) = 1$ which means that $r^*$ must solve the equation
\begin{align*}
    \left(1-r^*\right)^{n-1} +  \log\left(r^*\right) + \sum_{k=1}^{n-1} \frac{(1-r^*)^k}{k} = 0
\end{align*}


We can check that the left hand side as a function of $r^*$ is increasing on $r^* < 1/n$, is decreasing on $r^* > 1/n$, goes to $-\infty$ as $r^* \downarrow 0$, and goes to $0^+$ as $r^* \uparrow 1$. So this equation has a unique solution in $(0,1/n)$.

Lastly, we can derive the regret expression by direct calculation.\end{proof}


\section{Proofs and Discussions from Section~\ref{sec:extend-other-F}}\label{app:sec:extend-other-F} 

\begin{proof}[Proof of Proposition~\ref{prop:minimax-mix-eq-iid}]

If $\mathbf{F}$ is i.i.d.~with marginal $F$, then we will write $R_{n}(m,\mathbf{F}) = \tilde{R}(m,F)$. Then, if $\mathbf{F}$ is a mixture of i.i.d.~distributions $G \sim \mu$, we have $R_{n}(m,\mathbf{F}) = \bE_{G \sim \mu} \left[ R(m,G) \right]$. What we want to prove $\mathcal{R}_{n}(\mathcal{F}_{\text{mix}}) = \mathcal{R}_{n}(\mathcal{F}_{\text{iid}})$ is equivalent to 
\begin{align*}
    \inf_{m} \sup_{\mu}  \bE_{G \sim \mu} \tilde{R}(m,G) = \inf_{m} \sup_{F} \tilde{R}(m,F).
\end{align*}

For each $G$, $\tilde{R}(m,G) \leq \sup_{F} \tilde{R}(m,F)$, so $\bE_{G \sim \mu} \tilde{R}(m,G) \leq \sup_{F} \tilde{R}(m,F)$ for every $\mu$. Therefore, $\inf_{m} \sup_{\mu}  \bE_{G \sim \mu} \tilde{R}(m,G) \leq \inf_{m} \sup_{F} \tilde{R}(m,F)$.

Conversely, take any $\epsilon > 0$. By the definition of sup, for any $m$ there exists a $\hat{F}_m$ such that $\sup_{F} \tilde{R}(m,F) \leq \tilde{R}(m,\hat{F}_m) + \epsilon$. Therefore, $\inf_{m} \sup_{F} \tilde{R}(m,F) \leq \inf_{m} \tilde{R}(m,\hat{F}_m) + \epsilon \leq \inf_{m} \sup_{\mu} \bE_{G \sim \mu} \tilde{R}(m,G) + \epsilon$, where the last inequality holds because we can take $\mu$ to be the distribution that puts all the weight on $\hat{F}_m$. Since $\epsilon > 0$ is arbitrary, our desired equality is proved.\end{proof}

\begin{proof}[Proof of Proposition~\ref{prop:minimax-regret-arb}]
We will first show that $R_{n}(m^*,\mathbf{F}) \leq R_{n}(m^*,\mathbf{F}^{*})$. By symmetry, $R_{n}(m^*,\mathbf{F}^{*,k}) = R_{n}(m^*,\mathbf{F}^{*,1})$, so $R_{n}(m^*,\mathbf{F}^*) = \bE_{\mathbf{v} \sim \mathbf{F}^*}  R_{n}(m^*,\mathbf{v}) = \frac{1}{n} \sum_{k=1}^{n} R_{n}(m^*,\mathbf{F}^{*,k}) = R_{n}(m^*,\mathbf{F}^{*,1})$.
Therefore, equivalently we need to show that $R_{n}(\Phi^*,\mathbf{v}) \leq 1/e$ and that equality holds when $\mathbf{v}$ is on the support of $\mathbf{F}^{*,1}$. We have already derived the pointwise regret expression:
\begin{align*}
    R_{n}(\Phi^*,\mathbf{v}) =  v^{(1)} - v^{(1)} \Phi^*(v^{(1)}) + \int_{p = v^{(2)}}^{p=v^{(1)}} \Phi^*(p) dp
\end{align*}
We consider 3 cases.

Case 1. $v^{(1)} \geq 1/e \geq v^{(2)}$, then 
\begin{align*}
    R_{n}(\Phi^*,\mathbf{v}) =  v^{(1)} - v^{(1)} ( \log(v^{(1)}) + 1 ) + \int_{p=e^{-1/\lambda}}^{p=v^{(1)}} (\log(p) + 1) dp \\
    =  v^{(1)} - v^{(1)} (1 + \log(v^{(1)}) ) + \left[ ( p \log(p) - p ) + p \right]_{p=1/e}^{p=v^{(1)}} = 1/e
\end{align*}

Case 2. $v^{(1)} \geq v^{(2)} > 1/e$, then
\begin{align*}
    R_{n}(\Phi^*,\mathbf{v}) =  v^{(1)} - v^{(1)} ( \log(v^{(1)}) + 1 ) + \left[ ( p \log(p) - p ) + p \right]_{p=v^{(2)}}^{p=v^{(1)}} = -  v^{(2)} \log(v^{(2)})  < 1/e
\end{align*}
because the derivative of the last expression is $- \log(v^{(2)})-1 < 0$.

Case 3. $v^{(1)} < 1/e$, then
\begin{align*}
    R_{n}(\Phi^*,\mathbf{v}) =  v^{(1)} < 1/e
\end{align*}

Next, we show that $R_{n}(m^*,\mathbf{F}^*) \leq R_{n}(m,\mathbf{F}^*)$. We argue analogously to the proof of Theorem~\ref{thm:minimax-lambda-regret-main} that this is the same as showing that $m^*$ maximizes the expected revenue among all mechanisms. We also use the Myerson's payment characterization and do a similar calculation to get


\begin{align*}
    \bE_{\mathbf{v} \sim \mathbf{F}^*}\left[ \sum_{i=1}^{n} p_i(\mathbf{v}) \right] &= \frac{1}{n} \sum_{k=1}^{n} \bE_{v_k \sim F_k^*} \left[ p_k(v_k,\mathbf{v}_{-k}=\mathbf{0}) \right]
    \\&= \frac{1}{n} \sum_{k=1}^{n} \left( r^* x_k(v_k=1,\mathbf{v}_{-k}=\mathbf{0}) - \int_{\tilde{v}_k=0}^{\tilde{v}_k=r^*} x_k(v_k=\tilde{v}_k,\mathbf{v}_{-k}=\mathbf{0}) dv_k \right)
\end{align*}
Therefore, the expected revenue is maximized if and only if for every $k$, $x_k(v_k=1,\mathbf{v}_{-k}=\mathbf{0}) = 1$ and $x_k(v_k=\tilde{v}_k,\mathbf{v}_{-k}=\mathbf{0}) = 0$ for $\tilde{v}_k < r^*$, and $m^*= \text{SPA}(\Phi^*)$ satisfies this condition.\end{proof}


\begin{proof}[Proof of Proposition~\ref{lem:aff-mix}] 

To prove the result, we construct explicit distributions that fall in one class but not the other.
 
We will denote a distribution that takes one of the $m$ values $\{1,\ldots,m\}$ with an $m$-tuple, say, $\mathbf{p} = (p_1,\dots,p_m)$ is a distribution that takes value $k$ with probability $p_k$ for $1 \leq k \leq m$. We represent the joint distribution of two agents' values as a matrix, say $(p_{i,j})_{1 \leq i,j \leq m}$ takes the value $(i,j)$ with probability $p_{i,j}$ for each $i,j$. Then, the joint distribution for two agents when each agent has i.i.d.~values following $p$ is the outer product $\mathbf{p}\mathbf{p}^\top$ whose $(i,j)$ element is $p_i p_j$.

\begin{example}[mixture of i.i.d.~but not affiliated]
Let $\theta = 1/2, \mathbf{p} = (3/4,1/8,1/8), \mathbf{q} = (1/4,1/2,1/4)$. The distribution
\begin{align*}
    \theta \mathbf{p} \mathbf{p}^\top + (1-\theta) \mathbf{q} \mathbf{q}^\top =
    \begin{pmatrix}
    5/16 & 7/64 & 5/64 \\
    7/64 & 17/128 & 9/128 \\
    5/64 & 9/128 & 5/128
    \end{pmatrix}
\end{align*}
is a mixture of i.i.d.~but is not affiliated because $(7/64)(9/128) < (5/64)(17/128)$.
\end{example}
\begin{example}[affiliated but not mixture of i.i.d.]\!\!\!\footnote{We thank Vijay Kamble for pointing out that the example in the previous version of the paper was incorrect.}
Consider an exchangeable distribution over 3 values and $n=2$ agents given by
\begin{align*}
    P \equiv \begin{pmatrix}
    p_{11} & p_{12} & p_{13} \\
    p_{12} & p_{22} & p_{23} \\
    p_{13} & p_{23} & p_{33}
    \end{pmatrix}\,.
\end{align*}
The matrix is symmetric because the distribution is exchangeable. We can check that the distribution is affiliated if and only if 
\begin{align*}
    p_{11} p_{22} &\geq p_{12}^2 \\
    p_{11} p_{23} &\geq  p_{13}p_{12} \\
    p_{11} p_{33} &\geq p_{13}^2 \\
    p_{12} p_{23} &\geq p_{13} p_{22} \\
    p_{12} p_{33} &\geq p_{13} p_{23} \\
    p_{22} p_{33} &\geq p_{23}^2 ,
\end{align*}
that is, every $2 \times 2$ submatrix of $P$ has nonnegative determinant.

Moreover, each matrix representing an i.i.d. distribution is positive semi-definite, so if this distribution is a mixture of i.i.d., then $P$ is also positive semi-definite.



Consider $(p_{11},p_{12},p_{13},p_{22},p_{23},p_{33}) = (7/512, 41/1024, 1/64, 1/8, 1/8)$. We can check that every $2 \times 2$ submatrix of $P$ has nonnegative determinant, so it is affiliated, but the $3 \times 3$ matrix has determinant $\det(P) = -73/2097152$, which is negative, so it is not positive-semidefinite. Because it is not positive semi-definite, it is also not a mixture of i.i.d. \end{example} \end{proof}


\section{Proofs and Discussions from Section~\ref{sec:char-regret-ratio}}\label{app:sec:char-regret-ratio} 

\begin{proof}[Proof of Proposition~\ref{prop:sup-regret-spa-r}]

We will first prove the result for $\mathcal{F}_{\text{iid}}$ and will later show how to modify the proof for $\mathcal{F}_{\text{mix}}$ and $\mathcal{F}_{\text{aff}}$. 
We define $F(v) = \Pr(v_i \leq v)$ and $F_{-}(p) = \Pr(v_i < v)$. We have
\begin{align*}
    R_{n}(\text{SPA}(r), F) &=  \bE[v^{(1)}] - \bE[\max(v^{(2)},r) \1(v^{(1)} \geq r) ] \\
    &= \bE[v^{(1)}] - \bE[v^{(2)} \1(v^{(2)} > r)] - \bE[r \1(v^{(2)} \leq r \leq v^{(1)})] .
\end{align*}

The first term is 
\begin{align*}
    \bE[v^{(1)}] = \int_{v \in [0,1]} \Pr(v^{(1)} > v) dv = 1 - \int_{v \in [0,1]} F(v)^n dv .
\end{align*}

The second term's calculation is analogous to that of Lemma~\ref{lem:first-term-calc-gen}. We have
\begin{align*}
    \bE[v^{(2)} \1(v^{(2)} > r) ] 
    &= \int_{v' \in (r,1]} v' dF^{(2)}(v') 
    = \int_{v' \in (r,1]} \int_{v \in [0,v')} dv dF^{(2)}(v')\\
    &= \int_{v \in [0,r]} \int_{v' \in (r,1]} dF^{(2)}(v')dv +  \int_{v \in (r,1]} \int_{v' \in (v,1]} dF^{(2)}(v') dv \\
    &= \int_{v \in [0,r]} (1 - F^{(2)}(r) ) dv + \int_{v \in (r,1]} (1 - F^{(2)}(v) ) dv \\
    &= 1 - r F^{(2)}(r) - \int_{v \in (r,1]} F^{(2)}(v) dv .
\end{align*}

The third term is
\begin{align*}
    \bE[r \1(v^{(2)} \leq r \leq v^{(1)})] = r \Pr(v^{(2)} \leq r \leq v^{(1)}) = r( \Pr(v^{(2)} \leq r) - \Pr(v^{(1)} < r)) = r(F^{(2)}(r) - F_{-}(r)^n)
\end{align*}
Together, we have
\begin{align*}
    \sup_{F \in \mathcal{F}_{\emph{iid}}([0,1]) } R_{n}(\text{SPA}(r), F) =  r F_{-}(r)^n -  \int_{v \in [0,r]} F(v)^n dv  + \int_{v \in (r,1]} n F(v)^{n-1} - n F(v)^n dv .
\end{align*}

We first assume that $n \geq 2$. Note that the integrand $n F(v)^{n-1} - n F(v)^n$ is increasing for $F(v) \leq \frac{n-1}{n}$ and is decreasing for $F(v) \geq \frac{n-1}{n}$.  To minimize $\int_{v \in [0,p]^n} F(v)^n dv $ we must have $F(v) = 0$ for $v  \in [0,r-\epsilon]$ for arbitrarily small $\epsilon > 0$, and to maximize $\int_{v \in (p,1]} n F(v)^{n-1} - n F(v)^n dv$, the only constraint we have is $F(v) \geq c$ so for $v \in (r,1]$ we set $F(v) = \frac{n-1}{n}$ if $c \leq \frac{n-1}{n}$ and $F(v) = c$ otherwise. Note that the sup over first case of $c \leq \frac{n-1}{n}$ is simply the second case with $c = \frac{n-1}{n}$. Because we take the sup over $F$, we can let $\epsilon \downarrow 0$ and get that the worst-case regret is $\sup_{c \in [\frac{n-1}{n},1] } rc^n +  (1-r)( nc^{n-1} - n c^n) $. 
Now, the derivative of this expression of $c$ is $nc^{n-2} \left[ (1-r)(n-1) - ((1-r)n-r) c \right]$. The expression in $\left[ \cdots \right]$ is linear in $c$. At $c = \frac{n-1}{n}$, the expression is $r \frac{n-1}{n} \geq 0$. At $c = 1$, the expression is $2r-1$. So if $r \geq \frac{1}{2}$, the first derivative is always $\geq 0$, so the maximum is achieved at $c = 1$ and the value is $r$. If $r \leq \frac{1}{2}$, the maximum is achieved at $c^* = \frac{(1-r)(n-1)}{(1-r)n-r} \in \left[ \frac{n-1}{n}, 1 \right]$ and the value is $\frac{(1-r)^n (n-1)^{n-1}}{((1-r)n-r)^{n-1}}$.

Now we deal with the case $n=1$. This expression reduces to $r F_{-}(r) - \int_{v \in [0,r]} F(v) dv + \int_{v \in (r,1]} 1-  F(v) dv$ which has a sup of $\sup_{c \in [0,1]} \lambda - 1 + rc + (1-r)(1- c) = \max(1-r,r)$ because it is linear in $c$, so it achieves the extrema at the endpoints. 

Now we want to choose the optimal $r$ to minimize the worst-case regret. For $n \geq 2$, the worst-case regret $r$ is increasing on $[\frac{1}{2},1]$, so it is sufficient to consider $r$ in $[0,\frac{1}{2}]$ with regret \newline $\exp\left( n \log(1-r) + (n-1)\log(n-1) - (n-1) \log(n-(n+1)r)\right)$. The first derivative of the expression within exp is $-\frac{n}{1-r} + \frac{(n-1)(n+1)}{n-(n+1)r} = \frac{(n+1)r-1}{(1-r)(n-(n+1)r)}$. So it is maximized at $r^* = \frac{1}{n+1}$ and substituting this in gives the optimal regret $\left( \frac{n}{n+1}\right)^{n}$. For $n = 1$, $\max(1-r,r)$ is minimized at $r^* = \frac{1}{2}$ with regret $\frac{1}{2}$ so it has the same expression.

Note also that the proof implies that the worst case distribution (more precisely a sequence of distributions that achieves the minimax as $\epsilon \downarrow 0$) is a two-point distribution with weights $c^*$ and $1-c^*$ at $r-\epsilon$ and 1, respectively, with $c^* = \frac{(1-r)(n-1)}{(1-r)n-r}$. (If $r = 0$, then we put weights $c^*$ and $1-c^*$ at 0 and 1.) If we choose the optimal $r^* = \frac{1}{n+1}$, then $c^* = \frac{n}{n+1}$.

Therefore, the results for $\mathcal{F}_{\text{iid}}$ is proved. With exactly the same argument as that used for the proof of Proposition~\ref{prop:minimax-mix-eq-iid} (namely, because $\mathcal{F}_{\text{mix}}$ is the convex hull of $\mathcal{F}_{\text{iid}}$, there is an optimal pure strategy for Nature), we get that the minimax values for $\mathcal{F}_{\text{mix}}$ and $\mathcal{F}_{\text{iid}}$ are equal. For $\mathcal{F}_{\text{aff}}$, the proof is similar to that of Proposition~\ref{prop:minimax-regret-aff}. We note that the same calculation gives the regret
\begin{align*}
    & r \mathbf{F}_n^{(1)}(r^-) -  \int_{v \in [0,r]} \mathbf{F}_n^{(1)}(v) dv + \int_{v \in (r,1]} ( \mathbf{F}_n^{(2)}(v) - \mathbf{F}_n^{(1)}(v) ) dv \\
    &\leq  r \mathbf{F}_n^{(1)}(r) + \int_{v \in (r,1]} ( \mathbf{F}_n^{(2)}(v) -  \mathbf{F}_n^{(1)}(v) ) dv \\
    &= r \mathbf{F}_n^{(1)}(r) + \int_{v \in (r,1]} ( n \mathbf{F}_{n-1}^{(1)}(v) - n \mathbf{F}_n^{(1)}(v) ) dv \\
    &\leq r \mathbf{F}_n^{(1)}(r) + \int_{v \in (r,1]} ( n \mathbf{F}_{n}^{(1)}(v)^{(n-1)/n} - n \mathbf{F}_n^{(1)}(v) ) dv \\
    &\leq \sup_{c \in [0,1]} \left[ rc^n + \sup_{z \in [c,1]} \left\{ nz^{n-1} - n z^n \right\} \right]
\end{align*}
The third line uses the equation $\mathbf{F}_n^{(2)}(v) = n \mathbf{F}_{n-1}^{(1)}(v) - (n-1) \mathbf{F}_n^{(1)}(v) $ that we already used in the proof of Proposition~\ref{prop:minimax-regret-aff}. The fourth line uses the inequality $\mathbf{F}_{n-1}^{(1)}(v) \leq \mathbf{F}_{n}^{(1)}(v)^{(n-1)/n}$ that we proved as Lemma~\ref{lem:F-n-key} in the same proof. The fifth line is a pointwise maximization over $c = \mathbf{F}_n^{(1)}(r)^{1/n}$ and $z$ denoting $\mathbf{F}_n^{(1)}(v)^{1/n}$ for $v > r$ which is $\geq c$. The last expression is the same expression that we got for $\mathcal{F}_{\text{iid}}$, so it has the same minimax regret.\end{proof}

\section{Discussion on Upper Bound Information} \label{app:subsec:model-discuss-upper-bound}

We would first like to clarify that we do not assume that we know the actual upper bound of the support of the true underlying distribution $\bf{F}$ (for one buyer, that would be $\bar{v}_F = \inf \{x:  \int_{[x,\infty)} dF(x) = 0\}$, but simply that we know an upper bound $b$ of that value, i.e., the decision-maker only knows that $\bar{v}_F \le b$. This allows for example for the actual distribution to have $\bar{v}_F = 0.1b$, $\bar{v}_F = 0.5b$ or $\bar{v}_F = 0.99 b$ and the decision-maker needs to design a mechanism that is oblivious to $\bar{v}_F$ and more generally to $F$.

Throughout the paper, we assume that the known upper bound on the distribution is $1$, that is, the support is $[0,1]$. This is purely a normalization for expositional convenience. If the support is $[0,b]$ for some known upper bound $b$, then the resulting minimax optimal regret is $b$ times the original minimax regret, and the optimal mechanism and the corresponding worst-case distribution are multiplicatively scaled by $b$. We formally record this observation in the following proposition.

\begin{proposition}\label{prop:upper-bound-b}
Let the support of the buyers' valuations be $[0,b]^{n}$ for a known $b > 0$, and $r_n^*, \Phi_n^*$ be defined as in Theorem~\ref{thm:minimax-lambda-regret-main}. Let $\tilde{m}_n^*$ be a second-price auction with reserve $br$, with random $r \sim \Phi_n^*$. Equivalently, $\tilde{m}_n^* = \textnormal{SPA}(\tilde{\Phi}_n^*)$ with
\begin{align*}
    \tilde{\Phi}_n^*(v) = \left( \frac{v/b}{v/b - r_n^*} \right)^{n-1} \log\left( \frac{v/b}{r_n^*} \right) - \sum_{k=1}^{n-1} \frac{1}{k} \left( \frac{v/b}{v/b - r_n^*} \right)^{n-1-k} \text{ for } v \in [b r_n^*, b].
\end{align*}

Let $\tilde{F}_n^*$ be an isorevenue starting at $b r_n^*$ given by $\tilde{F}_n^*(v) = 1-r_n^*/(v/b)$ for $v \in [br_n^*, b)$ and $\tilde{F}_n^*(b) = 1$. 
Then, $(\tilde{m}_n^*, \tilde{F}_n^*)$ is a saddle point of the regret $R_n(m,F)$, and the minimax regret is $b$ times the original minimax regret, namely,
\begin{align*}
    b \left[ (1-r_n^*)^{n-1} - \int_{v=r_n^*}^{v=1} \left( 1 - \frac{r_n^*}{v} \right)^{n-1} dv \right].
\end{align*}

Note also that $(\tilde{m}_n^*, \tilde{F}_n^*)$ is still a saddle point if we instead assume that the support is $[a,b]$ with $a \leq r_n^* b$.
\end{proposition}

The motivation of the known support information assumption is that we operate in a world with minimal or no data, and significant ``Knightian'' uncertainty. The bounds need not be learned from data but rather are derived from asking experts, using domain knowledge, or common sense. \jaedit{Some examples might be launching a new product, or auctioning rarely traded goods (fine art, collectibles, jewelry).} \jacomment{should we add this example sentence? not quite sure what the examples should be} In these contexts, it is easier and more intuitive to come up with a reasonable \textit{range of values} than to guess something like the shape of the valuation distribution (either parametric or nonparametric like regular or monotone hazard rate) or distributional parameters like the mean or the optimal monopoly price.

\jacomment{Comments about other references from REviewer 2 I don't think we need to say that here.}

We will now show that even if the upper bound is misspecified, the performance loss (increase in minimax regret) is small and degrades gracefully with the error. In particular, we will compute the minimax regret when the upper bound assumed for the computation of the mechanism (as given in Proposition~\ref{prop:upper-bound-b}) is $b_{\textnormal{mech}}$, but the worst case distribution class $\mathbf{F}$ is over distributions with support $[0,b_{\textnormal{actual}}]$. Note that if $b_{\textnormal{mech}} = 1$, then the mechanism is $\textnormal{SPA}(\Phi_n^*)$.

\begin{proposition}\label{prop:worst-case-different-b}
Suppose that we use the mechanism $\textnormal{SPA}(\Phi_n^*)$, but the worst-case distribution class $\mathcal{F}$ is over support $[0,b]$. For $r_n^* \leq b \leq 1$, the worst case $F$ is
\begin{align*}
    F^*(v) = \begin{cases}
    0 &\text{ if } v \in [0,r_n^*] \\
    1-\frac{r_n^*}{v} &\text{ if } v \in [r_n^*,b) \\
    1 &\text{ if } v = b.
    \end{cases}
\end{align*}
For $b \geq 1$, the worst case $F$ is
\begin{align*}
    F^*(v) = \begin{cases}
    0 &\text{ if } v \in [0,r_n^*] \\
    1-\frac{r_n^*}{v} &\text{ if } v \in [r_n^*,s^*] \\
    1-\frac{r_n^*}{s^*} &\text{ if } v \in [s^*,b) \\
    1 &\text{ if } v = b
    \end{cases}
\end{align*}
where $s^* \in [nr_n^*,1]$ is a solution to \begin{align*}
    (b-1)  \left( 1 - \frac{n r_n^*}{s^*} \right) = \int_{v=s^*}^{v=1} (\mu^*)'(v) dv = \int_{v=s^*}^{v=1}  \Phi_n^*(v) \left[ \left( 1 - \frac{r_n^*}{s^*} \right) \left( \frac{v-nr_n^*}{v-r_n^*} \right) - 1 + \frac{n r_n^*}{s^*} \right] dv.
\end{align*}

\end{proposition}

The worst-case regret calculation in Proposition~\ref{prop:worst-case-different-b} assumes $b_{\textnormal{mech}} = 1$ for convenience, but by the same scaling argument used in Proposition~\ref{prop:upper-bound-b} earlier, we can compute the worst-case regret for any $(b_{\textnormal{mech}}, b_{\textnormal{actual}})$.

\begin{corollary}\label{corollary:worst-case-different-b}
If the worst-case regret under the mechanism $\textnormal{SPA}(\Phi_n^*)$ and support $[0,b]$ (that is, $b_{\textnormal{mech}} = 1, b_{\textnormal{actual}} = b$) as computed in Proposition~\ref{prop:worst-case-different-b} is given by a function $g(b)$, then the worst-case regret if the $b$ assumed by the seller is $b_{\textnormal{mech}}$ but the true $b$ is $b_{\textnormal{actual}}$ is given by $b_{\textnormal{mech}} g(b_{\textnormal{actual}}/b_{\textnormal{mech}})$.
\end{corollary}

We can now perform \textit{sensitivity analysis} on $b_{\textnormal{mech}}$ by computing the performance loss when the upper bound is misspecified. We fix Nature's distribution class to have support $[0,1]$, but we assume that the mechanism is computed assuming the potentially misspecified $b_{\textnormal{mech}}$ and evaluate the minimax regret of that mechanism as a function of $b_{\textnormal{mech}}$. The results are shown in Figure~\ref{fig:minimax-regret-vary-b-mech}. We can see that the minimax regret is Lipschitz in $b_{\textnormal{mech}}$ and that the performance guarantee is not very sensitive to $b_{\textnormal{mech}}$, especially for $n \geq 2$. For example, when $n = 2$, if the seller underestimates (resp. overestimates) the true upper bound by 20\%, the increase in regret is only 7.3\% (resp. 16.3\%).

\begin{figure}[h!]
    	\centering
    	\includegraphics[height=0.45\linewidth]{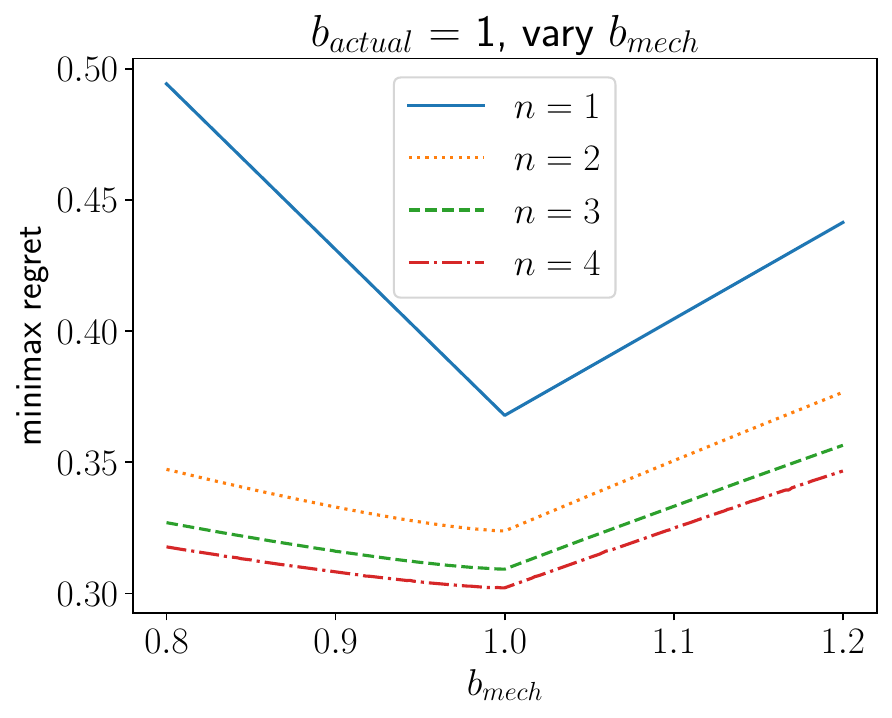}

	\caption{Minimax regret for each $n$ as a function of $b_{\textnormal{mech}}$ when Nature's worst-case distribution class has support $[0,b_{\textnormal{actual}}]$ with $b_{\textnormal{actual}}=1$, while the mechanism is computed assuming $b_{\textnormal{mech}}$.}
\label{fig:minimax-regret-vary-b-mech}
\end{figure}

\subsection{Proofs for Appendix \ref{app:subsec:model-discuss-upper-bound}}



\begin{proof}[Proof of Proposition~\ref{prop:worst-case-different-b}]
The regret under $\textnormal{SPA}(\Phi_n^*)$ and arbitrary $F$ supported on $[0,b]$ is given by
\begin{align*}
    r_n^*   - \int_{v=0}^{v=r_n^*} F(v)^n dv  + \int_{v=r_n^*}^{v=b} (1 -v {\Phi_n^*}'(v) - \Phi_n^*(v) ) (1-F(v)^n ) +  \Phi_n^*(v) n F(v)^{n-1}(1-F(v)) dv
\end{align*}

This is simply a modification of the (Regret-$F$) expression. Note also that $\Phi_n^*(v) = 1$ for $v \in [1,b]$, so we can also write the above regret as
\begin{align*}
    r_n^*   - \int_{v=0}^{v=r_n^*} F(v)^n dv  + \int_{v=r_n^*}^{v=1} (1 -v {\Phi_n^*}'(v) - \Phi_n^*(v) ) (1-F(v)^n ) +  \Phi_n^*(v) n F(v)^{n-1}(1-F(v)) dv \\ + \int_{v=1}^{v=b}  n F(v)^{n-1}(1-F(v)) dv
\end{align*}
for $b \geq 1$.

We first consider the case $n \geq 2$. If we maximize this regret expression pointwise, then we want $F^*(v) = 0$ for $v \in (0,r_n^*)$ to make the second term $\int_{0}^{r_n^*} F(v)^{n} dv$ zero. For the term that is the integral from $r_n^*$ to $\min(b,1)$, the first-order condition gives
\begin{align*}
    (1 -v {\Phi_n^*}'(v) - \Phi_n^*(v) ) ( -n F^*(v)^{n-1} ) +  \Phi_n^*(v) n ( (n-1)F^*(v)^{n-2} - n F^*(v)^{n-1} ) = 0  \\
    F(v) = \frac{(n-1) \Phi_n^*(v)}{1 -v {\Phi_n^*}'(v) +(n-1) \Phi_n^*(v)}
\end{align*}
Using the expression for ${\Phi_n^*}'(v)$ from (ODE-$\Phi$) reduces the above expression to
\begin{align*}
    F^*(v) = 1-\frac{r_n^*}{v}
\end{align*}
This is not surprising, because the argument in the main proof shows that this isorevenue $F^*(v)$ maximizes the regret pointwise, and we simply change the range.

For $b \geq 1$ there is also the term that is the integral from $1$ to $b$, which is $n F(v)^{n-1} (1-F(v))$, which is pointwise-optimized at $F^*(v) = \frac{n-1}{n}$.

This pointwise maximization gives a valid $F^*$ (and hence, the true worst-case $F$) if and only if the resulting $F^*$ is weakly increasing (that is, non-decreasing). If $b \leq 1$, then this gives a valid $F$ as shown in the proposition, and we are done. If $b > 1$, then the pointwise optimization gives $\lim_{v \uparrow 1} F^*(v) = 1 - r_n^* > \frac{n-1}{n} = \lim_{v \downarrow 1} F^*(v)$, so the resulting $F$ is decreasing at $1$, which is not feasible. Rather, we solve for the optimal $F$ by dualizing the monotonicity constraint and characterizing optimal Lagrangian multipliers. We still set $F^*(v) = 0$ for $v \in [0,r_n^*]$, but now we let $\mu(v)$ be the value of the dual function corresponding to the increasing constraint of $F$ at $v$, for $v \in [r_n^*,b]$. It is sufficient to give a specific $F^*$ and $\mu^*$ such that all duality-based constraints hold, namely, (i) $\mu^*(v) \geq 0$ for all $v \in [r_n^*,b]$, (ii) Lagrangian optimality:
\begin{align*}
    F^* \in \arg\max_{F} \mathcal{L}(F,\mu^*) := r_n^*    + \int_{v=r_n^*}^{v=1} (1 -v {\Phi_n^*}'(v) - \Phi_n^*(v) ) (1-F(v)^n ) +  \Phi_n^*(v) n F(v)^{n-1}(1-F(v)) dv \\ + \int_{v=1}^{v=b}  n F(v)^{n-1}(1-F(v)) dv  + \int_{v=r_n^*}^{v=b} \mu^*(v) dF(v) ,
\end{align*}
where the maximization is now over arbitrary $F$ in $[0,1]$ because the increasing constraint has been dualized away, and (iii) complementary slackness: $\mu^*(v) = 0$ if $F^*$ is strictly increasing at $v$ for every $v$. We will consider $F^*$ that is of the form given in the proposition, that is, $F^*(v) = 1-r_n^*/v$ isorevenue from $v \in [r_n^*,s^*]$, and $F^*(v) = 1-r_n^*/s^*$ constant from $v \in [s^*,b]$. Note that because $F^*(v)$ is strictly increasing on $[r_n^*,s^*]$, we must have $\mu^*(v) = 0$ for $v \in [r_n^*,s^*]$. Because $\lim_{v \uparrow b} F^*(v) = 1-r_n^*/s < 1 = F^*(b)$, we also set $\mu^*(b) = 0$. We will now use the Lagrangian optimality condition to pin down $\mu^*$ and show that it is indeed nonnegative. Using integration by parts,
\begin{align*}
    \int_{v=r_n^*}^{v=b} \mu^*(v) dF(v) = \int_{v=s}^{v=b} \mu^*(v) dF(v) = \mu(b) F(b) - \mu(s) F(s) - \int_{v=s}^{v=b} F(v) (\mu^*)'(v) dv \\=  - \int_{v=s}^{v=b} F(v) (\mu^*)'(v) dv,
\end{align*}
because $\mu^*(s) = \mu^*(b) = 0$.

Because the Lagrangian optimality condition is now optimizing over arbitrary $F$ in $[0,1]$, it is equivalent to a pointwise maximization condition. For $v \in [r_n^*,s^*]$, $F^*(v) = 1 - r_n^*/v$ satisfies the pointwise maximization with $\mu^*(v) = 0$ as we have shown. For $v \in (s^*,1)$ and $v \in (1,b)$ we have:
\begin{align*}
    &\text{ for } v\in (s^*,1), \quad &&F^*(v) \in \arg\max_{F(v) \in [0,1]} (1 -v {\Phi_n^*}'(v) - \Phi_n^*(v) ) (1-F(v)^n) \\ & &&\quad\quad\quad + \Phi_n^*(v) n F(v)^{n-1} (1-F(v)) - F(v) (\mu^*)'(v) , \\
    &\text{ for } v\in (1,b), \quad &&F^*(v) \in \arg\max_{F(v) \in [0,1]}  n F(v)^{n-1} (1-F(v)) - F(v) (\mu^*)'(v) 
\end{align*}
First-order conditions give
\begin{align*}
    (1 -v {\Phi_n^*}'(v) - \Phi_n^*(v) ) (1-F(v)^n) + \Phi_n^*(v) n F(v)^{n-1} (1-F(v)) -  (\mu^*)'(v) = 0 \text{ for } v \in (s^*,1) \\
    n ((n-1)F^*(v)^{n-2} - n F^*(v)^{n-1} ) - (\mu^*)'(v) = 0 \text{ for } v \in (1,b)
\end{align*}
We note that $F^*(v) = F^*(s) = 1-r_n^*/s$ for all $v \in (s^*,1)$ and $v \in (1,b)$. We also have an expression for ${\Phi_n^*}'(v)$ in terms of $\Phi_n^*(v)$ from (ODE-$\Phi$). We then have, for $v \in (s^*,1)$,
\begin{align*}
    (\mu^*)'(v) = n F^*(s)^{n-2} \Phi_n^*(v) \left[ F^*(s) \left( \frac{v-nr_n^*}{v-r_n^*} \right) + n-1-n F^*(s) \right] \\ \geq n F^*(s)^{n-2} \Phi_n^*(v) \left[ F^*(s) \left( \frac{s^*-nr_n^*}{s^*-r_n^*} \right) + n-1-n F^*(s) \right] =0
\end{align*}
where the inequality is true because $v \geq s^* \geq r_n^*$ and the function $\frac{v-nr_n^*}{v-r_n^*}$ is increasing in $v$, and we can directly check the last equality by substituting $F^*(s) = 1-r_n^*/s$. 

For $v \in (1,b)$, we have
\begin{align*}
    (\mu^*)'(v) =  n F^*(s^*)^{n-2} (n-1-n F^*(s^*)) = n F^*(s^*)^{n-2} \left( \frac{n r_n^*}{s^*} - 1 \right)
\end{align*}
If we can prove that $s^*\geq n r_n^*$, then we will have $(\mu^*)'(v) \leq 0$ for $v \in (1,b)$. Combining this with $(\mu^*)'(v) \geq 0$ for $v \in (s^*,1)$ and $\mu^*(s) = \mu^*(b) = 0$ gives us that $\mu^*$ is increasing from $s^*$ to $1$ and decreasing from $1$ to $b$, being nonnegative the entire time, which verifies (i). These constraints also pin down $s^*$ from
\begin{align*}
    0 = \mu^*(s^*) - \mu^*(b) = \int_{v=s^*}^{v=1} (\mu^*)'(v) dv + \int_{v=1}^{v=b} (\mu^*)'(v) dv 
\end{align*}
which is equivalent to
\begin{align*}
    (b-1)  \left( 1 - \frac{n r_n^*}{s^*} \right) = \int_{v=s^*}^{v=1} (\mu^*)'(v) dv = \int_{v=s^*}^{v=1}  \Phi_n^*(v) \left[ \left( 1 - \frac{r_n^*}{s^*} \right) \left( \frac{v-nr_n^*}{v-r_n^*} \right) - 1 + \frac{n r_n^*}{s^*} \right] dv.
\end{align*}


It remains to show that the above equation has a solution in $s^* \in [nr_n^*,1]$. We consider the difference between the left hand side and the right hand side as a function in $s^*$. For $s^* = nr_n^*$, the left hand side is zero while the right hand side is nonnegative, so the difference is nonpositive. For $s^* = 1$, the left hand side is nonnegative while the right hand side is zero, so the difference is nonpositive. By the intermediate value theorem, there exists a $s^* \in [nr_n^*,1]$ such that the difference is zero, and we are done.

Lastly, we consider the case $n  = 1$. For $1/e \leq b \leq 1$, the regret is identically $1/e$ under $\Phi_n^*$. For $b \geq 1$, the regret is $1/e + \int_{v=1}^{v=b} (1-F(v)) dv$, which is maximized at $F(v) = 0$ for $v \in (1,b)$, giving the worst-case regret of $1/e+b-1$. This is equivalent to $s^* = r_n^* = 1/e$, which is indeed a solution to the equation in the proposition.\end{proof}

\section{Further Discussions About The Model}\label{app:sec:model-discuss}

\subsection{Discussions on the Benchmark}\label{app:subsec:model-discuss-benchmark}

\jacomment{Is this equivalence true for exchangeable distributions?}

In the main text, we consider the benchmark to be the maximum revenue achievable when the \textit{valuations} $\mathbf{v}$ of the buyers are known, namely, $\max(\mathbf{v})$. This is also known as the \textit{first-best} benchmark. There is another benchmark that we consider, the \textit{second-best} benchmark, defined as the maximum revenue achievable when the \textit{distribution} $\mathbf{F}$ of the buyers are known. We denote this benchmark by $\text{OPT}(\mathbf{F})$. When $\mathbf{F} = F^n$ is i.i.d., we abuse the notation and write $\textnormal{OPT}(\mathbf{F})$ as $\textnormal{OPT}(F)$.

Both the first-best benchmark and the second-best benchmark are extensively used in the literature, and there is not necessarily a single universal benchmark. The first-best benchmark has been used in the economics, operations, and computer science literatures such as \citep{BergemannSchlag08,CaldenteyLiuLobel17,KocyigitKR20-new,robust-monopoly-regulation,kleinberg-yuan}. 

The first-best benchmark is also reminiscent of the \textit{offline optimum} benchmark, which is extensively used in the analysis of algorithms \citep{BorodinElYaniv}. An advantage of the first-best benchmark is that it is easily computable and can be evaluated counterfactually (using the reported values when the mechanism is DSIC).

The main technical result of this section is that when the distribution class consists of all distributions ($\mathcal{F} =\mathcal{F}_{\textnormal{all}}$), minimax regret against the first-best benchmark and the second-best benchmark are equal, so the choice of the benchmark is immaterial in that case. 

\begin{proposition}\label{prop:fb-sb-equiv}
If $\mathcal{F} = \mathcal{F}_{\textnormal{all}}$, then
\begin{align*}
    \inf_{m \in \mathcal{M}} \sup_{\mathbf{F} \in \mathcal{F}} \bE_{\mathbf{v} \sim \mathbf{F}} \left[ \textnormal{OPT}(\mathbf{F})   - \sum_{i=1}^{n} p_i(\mathbf{v}) \right] = \inf_{m \in \mathcal{M}} \sup_{\mathbf{F} \in \mathcal{F}} \bE_{\mathbf{v} \sim \mathbf{F}} \left[ \max(\mathbf{v}) - \sum_{i=1}^{n} p_i(\mathbf{v}) \right]
\end{align*}
\end{proposition}

\begin{proof}[Proof of Proposition~\ref{prop:fb-sb-equiv}]

It is clear that $\textnormal{OPT}(\mathbf{F}) \leq \bE_{\mathbf{v} \sim \mathbf{F}} [ \max(\mathbf{v}) ]$, so 
\begin{align*}
    \inf_{m \in \mathcal{M}} \sup_{\mathbf{F} \in \mathcal{F}} \bE_{\mathbf{v} \sim \mathbf{F}} \left[ \textnormal{OPT}(\mathbf{F})   - \sum_{i=1}^{n} p_i(\mathbf{v}) \right] \leq \inf_{m \in \mathcal{M}} \sup_{\mathbf{F} \in \mathcal{F}} \bE_{\mathbf{v} \sim \mathbf{F}} \left[ \max(\mathbf{v}) - \sum_{i=1}^{n} p_i(\mathbf{v}) \right] .
\end{align*}

It remains to prove the opposite inequality.

For each $\mathbf{v} \in [0,1]^{n}$, let $\delta_{\mathbf{v}}$ be a (joint) distribution that puts all weight on $\mathbf{v}$. Note that $\mathbf{v} \in \mathcal{F}_{\textnormal{all}}$. Therefore, for any fixed $m$ and $\mathbf{v}$, the sup over all $\mathbf{F}$ must be at least the value evaluated at $\mathbf{F} = \delta_{\mathbf{v}}$, namely,
\begin{align*}
    & \sup_{\mathbf{F} \in \mathcal{F}} \bE_{\mathbf{v} \sim \mathbf{F}} \left[ \textnormal{OPT}(\mathbf{F})   - \sum_{i=1}^{n} p_i(\mathbf{v}) \right] \geq      \left[ \textnormal{OPT}(\delta_{\mathbf{v}}) - \sum_{i=1}^{n} p_i(\mathbf{v}) \right] =    \left[ \max(\mathbf{v}) - \sum_{i=1}^{n} p_i(\mathbf{v}) \right]
\end{align*}
Since this is true for every $\mathbf{v}$, we have
\begin{align*}
    \sup_{\mathbf{F} \in \mathcal{F}} \bE_{\mathbf{v} \sim \mathbf{F}} \left[ \textnormal{OPT}(\mathbf{F})   - \sum_{i=1}^{n} p_i(\mathbf{v}) \right] \geq \sup_{\mathbf{v} \in [0,1]^{n} } \left[ \max(\mathbf{v}) - \sum_{i=1}^{n} p_i(\mathbf{v}) \right]  =  \sup_{\mathbf{F} \in \mathcal{F}} \bE_{\mathbf{v} \sim \mathbf{F}} \left[ \max(\mathbf{v}) - \sum_{i=1}^{n} p_i(\mathbf{v}) \right].
\end{align*}
Since this is true for every $m \in \mathcal{M}$, we have
\begin{align*}
    \inf_{m \in \mathcal{M}} \sup_{\mathbf{F} \in \mathcal{F}} \bE_{\mathbf{v} \sim \mathbf{F}} \left[ \textnormal{OPT}(\mathbf{F})   - \sum_{i=1}^{n} p_i(\mathbf{v}) \right] \geq \inf_{m \in \mathcal{M}} \sup_{\mathbf{F} \in \mathcal{F}} \bE_{\mathbf{v} \sim \mathbf{F}} \left[ \max(\mathbf{v}) - \sum_{i=1}^{n} p_i(\mathbf{v}) \right],
\end{align*}
which concludes the proof.\end{proof}

Extending our work to the characterization of the minimax regret against the second-best benchmark when $\mathcal{F} = \mathcal{F}_{\textnormal{iid}}$ is an important direction for future work. In general, it is not clear a priori if it is possible to solve for the optimal minimax mechanism. In this paper, we were able to show that it is possible to do so against the first-best benchmark through a pure saddle point approach. We just want to highlight that characterizing a minimax mechanism against the second-best benchmark would require a different approach and is a priori substantially more difficult. In particular, a pure saddle point $(m^*,F^*)$ cannot exist. Indeed, if it did, letting $\text{OPT}(F)$ be the second-best benchmark and $\text{Rev}(m,F)$ the revenue of mechanism $m$ under distribution $F$, one could swap the min and the max in the zero-sum game, giving minimax regret
\begin{align*}
\min_{m} \max_{F} [ \text{OPT}(F) - \text{Rev}(m,F) ] = \max_{F} \min_{m} [ \text{OPT}(F) - \text{Rev}(m,F) ] = 0,
\end{align*}
where the last equality follows from $\max_{m} \text{Rev}(m,F) = \text{OPT}(F)$ for any $F$. This leads to a contradiction. Instead, we would need to identify a \textit{mixed} saddle point $(m^*, D^*)$ where $D^*$ is a \textit{distribution over distributions}, which is, at least at this stage, a much harder object to deal with.

Viewed through this lens, our proof of Proposition~\ref{prop:fb-sb-equiv} amounts to showing that if $(m^*,\mathbf{F}^*)$ is a saddle point of the first-best minimax regret problem, then $(m^*,D^*)$ is a saddle point of the second-best problem, where $D^*$ is a distribution over all point distributions $\delta_{\mathbf{v}}$ such that the mixture of $D^*$ is $\mathbf{F}^*$ (i.e., the weight on $\delta_{\mathbf{v}}$ is the measure of $\mathbf{F}^*$ at $\mathbf{v}$). This is true because for any distribution $\delta_{\mathbf{v}}$ on the support of $D^*$, the first-best and the second-best benchmark are the same, so the saddle inequalities for the first-best problem immediately transfer to the second-best problem.

\subsection{Discussion on Minimax Regret as Value of Competition and Limiting Minimax Regret}\label{app:subsec:model-discuss-limit-regret}

Lastly, we note that interpreting the minimax regret against the first-best benchmark as quantifying the value of competition is subtle. We note for example, that in previous works by \cite{KocyigitKR20-new,KocyigitRK21-old}, while nominally about multiple bidders, the worst-case ended up reducing back the problem of one-bidder performance. Given the power of nature, the decision-maker could never leverage additional competition to improve performance.   In our work, we limit the power of nature (e.g., to i.i.d values) and this allows to capture the improved worst-case performance as a function of the  number of bidders $n$.

The fact that the limiting regret as $n \to \infty$ is strictly positive is surprising at first glance (we do not think it is ex ante obvious), but it does make some sense once we realize that having infinitely many bidders is not the same as perfect knowledge because the seller has to commit to a selling mechanism beforehand and is restricted in how the mechanism can use that information to extract revenue. This apparently counterintuitive result is driven by the fact that Nature can choose a distribution that depends on $n$. The result is robust to the choice of the benchmark and is not simply an artifact of our modeling assumptions. The following proposition shows that for any $n$, the minimax regret against the second-best benchmark is at least $e^{-2} > 0$.

\begin{proposition}\label{prop:minimax-regret-opt-limit}
For $\mathcal{F} \in \{ \mathcal{F}_{\textnormal{iid}}, \mathcal{F}_{\textnormal{aff}}, \mathcal{F}_{\textnormal{mix}}, \mathcal{F}_{\textnormal{exc}}, \mathcal{F}_{\textnormal{all}} \}$, we have for all $n\ge 1$,
\begin{align*}
    \inf_{m \in \mathcal{M}} \sup_{\mathbf{F} \in \mathcal{F}} \left\{ \textnormal{OPT}(\mathbf{F}) - \bE_{\mathbf{v} \sim \mathbf{F}} \left[ \sum_{i=1}^{n} p_i(\mathbf{v}) \right] \right\} \ge e^{-2}.
\end{align*}
\end{proposition}

\begin{proof}[Proof of Proposition~\ref{prop:minimax-regret-opt-limit}] As before, it is sufficient to show the result for the smallest class $\mathcal{F}_{\textnormal{iid}}$. 

Fix $q = 1/n$, and consider the two-point distribution $F_{r}$ that takes values $r$ with probability $q$ and zero otherwise. If the distribution $F_r$ is known to the seller, the optimal mechanism against any such a distribution is a SPA with reserve price $r$:
\begin{align*}
    \textnormal{OPT}(F_r) = r \prob_{\mathbf{v} \sim F_{r}^n}(v^{(1)}=r) .
\end{align*}
Also, by the IR constraint, under $F_r$, the payment is always bounded above by $r$, the maximum possible valuation, and if $\mathbf{v} = \mathbf{0}$, the payment is bounded above by 0. Therefore,
\begin{align*}
    \bE_{\mathbf{v} \sim F_r^{n}} \left[ \sum_{i=1}^{n} p_i(v) \right] \leq r  \prob_{\mathbf{v} \sim F_{r}^n}(v^{(1)}=r, v^{(2)}=r)  + p_1(r,\mathbf 0_{-1}) \prob_{\mathbf{v} \sim F_{r}^n}(v^{(1)}=r, v^{(2)}=0) .
\end{align*}
Now consider a distribution-over-distributions $D$ in which we randomize over the values of $r$ according to the isorevenue distribution on $[1/e,1]$, that is 
\begin{align*}
    D(r) = \begin{cases}
    0 &\text{ if } r \in [0, 1/e] \\
    1 - \frac{1}{e r} &\text{ if } r \in [1/e,1) \\
    1 &\text{ if } r = 1.
    \end{cases}
\end{align*}
By Yao's Lemma we obtain
\begin{align*}
     & \inf_{m \in \mathcal{M}} \sup_{\mathbf{F} \in \mathcal{F}} \left\{ \textnormal{OPT}(\mathbf{F}) - \bE_{\mathbf{v} \sim \mathbf{F}} \left[ \sum_{i=1}^{n} p_i(\mathbf{v}) \right] \right\} \\
     &\quad \ge \inf_{m \in \mathcal{M}} \bE_{r \sim D} \left[ \textnormal{OPT}(F_r) - \bE_{\mathbf{v} \sim F_r^n} \left[ \sum_{i=1}^{n} p_i(\mathbf{v}) \right] \right] \\
     &\quad\ge 
     \inf_{m \in \mathcal{M}} \bE_{r \sim D} \left[
     r \prob_{\mathbf{v} \sim F_{r}^n}(v^{(1)}=r) - r  \prob_{\mathbf{v} \sim F_{r}^n}(v^{(1)}=r, v^{(2)}=r)  - p_1(r,\mathbf 0_{-1}) \prob_{\mathbf{v} \sim F_{r}^n}(v^{(1)}=r, v^{(2)}=0) \right]\\
     &\quad = n q (1-q)^{n-1} \inf_{m \in \mathcal{M}} \bE_{r \sim D} \left[
     r  - p_1(r,\mathbf 0_{-1}) \right]\\
     &\quad = n q (1-q)^{n-1} \left( \bE_{r \sim D} [r] - \sup_{m \in \mathcal{M}} \bE_{r \sim D} \left[p_1(r,\mathbf 0_{-1}) \right] \right)\,,   
\end{align*}
where the second inequality follows from the expresion for $\textnormal{OPT}(F_r)$ and the upper bound on the expected payment that we have previously derived, the first equality because $\prob_{\mathbf{v} \sim F_{r}^n}(v^{(1)}=r) = \prob_{\mathbf{v} \sim F_{r}^n}(v^{(1)}=r, v^{(2)}=r) + \prob_{\mathbf{v} \sim F_{r}^n}(v^{(1)}=r, v^{(2)}=0)$ and $\prob_{\mathbf{v} \sim F_{r}^n}(v^{(1)}=r, v^{(2)}=0) = n q (1-q)^{n-1}$, and the last from extracting constant terms.

Using the IC and IR constraints, we obtain that
\begin{align*}
    \sup_{m \in \mathcal{M}} \bE_{r \sim D} \left[p_1(r,\mathbf 0_{-1}) \right] 
     & \le \sup_{m \in \mathcal{M}} \bE_{r \sim D} \left[r x_1(r,\mathbf 0_{-1}) -  \int_{\tilde{v}_1=0}^{\tilde{v}_1=r} x_1(\tilde{v}_1,\mathbf{0}_{-1}) d\tilde{v}_1 \right]\\ 
     & = \sup_{m \in \mathcal{M}} \frac 1 e x_1(1, \mathbf 0_{-1}) + \int_{1/e}^1 r x_1(r,\mathbf 0_{-1}) \frac {1}{er^2} dr - \int_0^{1} x_1(r, \mathbf 0_{-1}) (1-D(r)) dr \\ 
     & = \sup_{m \in \mathcal{M}} \frac 1 e x_1(1, \mathbf 0_{-1})  - \int_0^{1/e} x_1(r, \mathbf 0_{-1}) dr \\   
     & \le \frac 1 e\,,
\end{align*}
where the first equation follows by integration by parts, and the last inequality because the allocation satisfies $0 \le x_1(\mathbf v) \le 1$. 

The result follows because $\bE_{r \sim D} [r] = 2/e$ and using that $n q (1-q)^{n-1} \ge 1/e$ for $q=1/n$.
\end{proof}

The result that the limiting regret against the second-best benchmark is strictly positive does not contradict the well-known result of \cite{BulowKlemperer96}. They show that the revenue of SPA with no reserve  is at least $(n-1)/n$ of the second-best benchmark (optimal revenue with knowledge of the distribution), \textit{assuming that the distribution is regular.}\footnote{More precisely, their Assumption A.1 assumes that the marginal revenue curve is downward-sloping. In the central case of i.i.d.~bidders, this assumption is equivalent to regularity.} Their result does not apply in our setting when one does not assume regularity. 


Even though the limiting regret is strictly positive, we think it is possible to meaningfully the minimax regret value as a function of $n$ against the first-best benchmark as measuring some notion of value of competition. We could view the $n = 1$ value as the ``ceiling'' and the $n = \infty$ value as the ``floor,'' where the competition is maximized and what remains is purely the effect of the benchmark itself. In a sense, we calculate the limiting minimax regret not for its own sake, but to establish a reference point. The value of competition for each $n$ is how much the regret has gone down from the ceiling compared to the floor, and the benchmark effect is implicitly subtracted away. This is how we interpret our result that the value of competition is significant for low $n$ (say from $n = 1 $ to $n=2$) but has diminishing returns.

\end{document}